%% file: main_arxiv.tex
\documentclass[
11pt,
paper=letter,
]{scrartcl}
\include{packages}
\begin{document}
	
	\title{On the Time and Space Complexity of ABA Prevention and Detection\footnote{This research was undertaken, in part, thanks to funding from the Canada Research Chairs program and from the Discovery Grants program of the Natural Sciences and Engineering Research Council of Canada (NSERC).}}
	\author{ Zahra Aghazadeh \\ University of Calgary
		\and Philipp Woelfel \\ University of Calgary}
	\date{}                
		
	\maketitle
	\input{abstract}

	\sloppy
	\input{intro}
	\input{lowerbounds}
	\input{aba_detection_ub}
	
	\bibliographystyle{abbrv}
	\bibliography{ref}
	\appendix
	\clearpage
	\appendix
	\input{appendix}
	\input{proof_of_ABA_reg_from_reg}

	\input{LLSC_from_CAS_proof}

\end{document}

%% file: packages.tex
\usepackage{comment}
\excludecomment{comment}
\usepackage[short]{pwbib}

\usepackage{enumitem}
\setlist{noitemsep,topsep=0pt,parsep=0pt,partopsep=0pt}

\usepackage{privmath}
\usepackage{fullpage}
\usepackage{
amsmath,
amsfonts,%
amssymb,%
color,
framed,
rotating 
} 
\usepackage{cite}
\usepackage{amsthm}
\usepackage[textsize=tiny, color=blue!30,disable]{todonotes} 
\usepackage{hyperref}

\DeclareMathOperator*{\indist}{\sim}
\DeclareMathOperator*{\nindist}{\not\sim}

\usepackage{tikz}
\usetikzlibrary{%
  arrows,%
  automata,%
  backgrounds,%
  calc,%
  decorations.pathreplacing,%
  decorations.pathmorphing,%
  decorations.markings,%
  fit,%
  patterns,%
  positioning,%
  shapes.arrows,
  shapes.callouts,
  topaths
}
\tikzstyle{config} = [draw,circle, inner sep=0pt, minimum height=2.5em,fill=red!20]
\tikzstyle{quiescent} = [config,double]
\tikzstyle{arrow} = [draw, -latex] 
\tikzstyle{indist} = [draw,dotted,shorten >=4pt, shorten <=4pt,inner sep=2pt,blue!50!black] 
\tikzstyle{indistlabel} = [anchor=center,blue!50!black,fill=white]
\tikzstyle{box} = [rectangle, draw, thin,fill=gray!25, text centered, align=center,font=\tiny, above=2pt]
\tikzstyle{scheduledescr} = [above,align=center]
\tikzstyle{schedulelabel} = [below=2pt]
\tikzstyle{block} = [rectangle, draw, fill=blue!25, text width=6em, text centered, minimum height=5ex]
\tikzstyle{decision} = [diamond, draw, fill=blue!25, text centered, inner sep=0pt, text width=5em, aspect=1.8]

\tikzset{
  yields/.style={arrow,decorate,decoration={snake,amplitude=.35mm,segment length=2mm,post length=3mm},shorten >=1pt,thin},  
  timearrow/.style={ 
    decoration={markings,mark=at position 1 with {\arrow[scale=1.3]{>}}},
    postaction={decorate},
    shorten >=0.4pt
  },
  time/.style={ellipse callout, callout pointer arc=5, callout relative pointer={(0,-.4)}, above=.4,
    draw=black!70,top color=white,bottom color=black!20, semithick},
  change/.style={rectangle callout, callout pointer shorten=1ex, callout pointer width=.4em, 
    align=left, callout relative pointer={(0,.5)},below=.5,
    draw,rounded corners, semithick},
  state/.style={cloud callout, callout pointer start size=5pt, callout pointer end size=1pt, callout pointer segments=4,
    cloud ignores aspect=true,cloud puffs=12,cloud puff arc=180,inner sep=0pt,
    align=left, callout relative pointer={(0,.4)},below=.5,
    draw,rounded corners,semithick}          
}

\usepackage[flushleft]{threeparttable}

\newcommand{\pwtodo}[1]{\todo[inline,color=red!20]{#1}}
\newcommand{\oldtodo}[1]{\todo[inline,color=green!20]{#1}}
\newcommand{\oldpwtodo}[1]{\todo[inline,color=green!20]{#1}}

\newcommand{\remove}[1]{}

\usepackage{framed}
\usepackage[
nokwfunc,noresetcount,vlined,boxruled]{algorithm2e} 
\DontPrintSemicolon
\SetAlgoSkip{}
\SetAlgoInsideSkip{}

\newcommand{\Class}[1]{\caption{\AlCapSty{Class} \FuncSty{#1}}}
\newcommand{\Method}[1]{\caption{\AlCapSty{\small Method} \FuncSty{#1}}}

\newenvironment{classfigure}[1][htb]{
	\begin{figure*}[#1]  \begin{leftbar} \begin{minipage}{\linewidth}%
			\NoCaptionOfAlgo\LinesNotNumbered%
		}{\end{minipage}  \end{leftbar} \end{figure*}}

\makeatletter

\newenvironment{method}[1]{%
	\vskip1ex
	\begin{minipage}{\textwidth}%
		\LinesNumbered%
		\@twocolumnfalse
		\begin{algorithm}[H]\Method{#1}
		}{\end{algorithm}\@twocolumntrue
	\end{minipage}}

		\makeatother
		
		\SetKw{shared}{shared}
		\SetKw{local}{local}

		\newcommand{\IlIf}[2]{\KwSty{if} #1 \KwSty{then} #2}

\SetKwFunction{Read}{Read}
\SetKwFunction{Write}{Write}
\SetKwFunction{DRead}{DRead}
\SetKwFunction{DWrite}{DWrite}
\SetKwFunction{CAS}{CAS}
\SetKwFunction{LL}{LL}
\SetKwFunction{SC}{SC}
\SetKwFunction{VL}{VL}
\SetKwFunction{contains}{contains}
\SetKwFunction{enq}{enq}
\SetKwFunction{deq}{deq}
\SetKwFunction{add}{add}
\SetKwFunction{rem}{remove}
\SetKwFunction{GetSeq}{GetSeq}
\SetKw{True}{True}
\SetKw{False}{False}
\SetKwFunction{success}{True}
\SetKwFunction{failure}{False}
\SetKwFunction{True}{True}
\SetKwFunction{False}{False}
\SetKwFunction{WWrite}{WeakWrite}
\SetKwFunction{WRead}{WeakRead}
\SetKwFunction{RMW}{RMW}
\SetKwFunction{TAS}{TAS}
\SetKwFunction{WReadL}{WeakRead$^\ast$}
\SetKwFunction{WWriteL}{WeakWrite$^\ast$}

\newcommand{\yields}[1]{\stackrel{#1}{\rightsquigarrow}}
\newcommand{\Exec}{\mathrm{Exec}}
\newcommand{\Conf}{\mathrm{Conf}}
\newcommand{\sig}{\ensuremath{\mathrm{sig}}}
\newcommand{\reg}{\ensuremath{\mathrm{reg}}}
\newcommand{\RR}{\mathcal{R}}
\newcommand{\CCov}{\mathrm{CCov}}
\newcommand{\WCov}{\mathrm{WCov}}

\newcommand{\val}{\mathrm{val}}
\newcommand{\FF}{\mathcal{F}}
\newcommand{\UU}{\mathcal{U}}

\newcommand{\lin}[1]{\ensuremath{\ell(#1)}}
\newcommand{\inv}[1]{\ensuremath{inv(#1)}}
\newcommand{\rsp}[1]{\ensuremath{rsp(#1)}}

\newcommand{\PP}{\mathcal{P}}

\SetCommentSty{mycommfont}

\usepackage[noabbrev,capitalise,nameinlink]{cleveref}
\crefname{line}{line}{lines}
\Crefname{line}{Line}{Lines}
\newcommand{\crefrangeconjunction}{--}
\crefname{page}{page}{pages}
\Crefname{page}{Page}{Pages}
\crefname{equation}{Eq.}{Eq.}
\Crefname{equation}{Eq.}{Eq.}

\newtheorem{theorem}{Theorem}
\newtheorem{lemma}{Lemma}
\newtheorem{corollary}{Corollary}
\newtheorem{claim}{Claim}
\newtheorem{observation}{Observation}



%% file: abstract.tex
\begin{abstract}
	  We investigate the time and space complexity of detecting and preventing ABAs in shared memory algorithms for systems with $n$ processes and bounded base objects.
	  To that end, we define ABA-detecting registers, which are similar to normal read/write registers, except that they allow a process $q$ to detect with a read operation, whether some process wrote the register since $q$'s last read.
	  ABA-detecting registers can be implemented trivially from a single unbounded register, but we show that they have a high complexity if base objects are bounded:
	  An obstruction-free implementation of an ABA-detecting single bit register cannot be implemented from fewer than $n-1$ bounded registers.
	  Moreover, bounded CAS objects (or more generally, conditional read-modify-write primitives) offer little help to implement ABA-detecting single bit registers: We prove a linear time-space tradeoff for such implementations.
	  We show that the same time-space tradeoff holds for implementations of single bit LL/SC primitives from bounded writable CAS objects.
	  This proves that the implementations of LL/SC/VL by Anderson and Moir \cite{AM1995a} as well as Jayanti and Petrovic \cite{JP2003a} are optimal.
	  
	  We complement our lower bounds with tight upper bounds: 
	  We give an implementation of ABA-detecting registers from $n+1$ bounded registers, which has step complexity $O(1)$. 
	  We also show that (bounded) LL/SC/VL can be implemented from a single bounded CAS object and with $O(n)$ step complexity.
	  Both upper bounds are asymptotically optimal with respect to their time-space product. 
	  
	  These results give formal evidence that the ABA problem is inherently difficult, that even writable CAS objects do not provide significant benefits over registers for dealing with the ABA problem itself, and that there is no hope of finding a more efficient implementation of LL/SC/VL from bounded CAS objects and registers than the ones from \cite{JP2003a,AM1995a}.
\end{abstract}

%% file: intro.tex
\section{Introduction}
\oldtodo{Check the paper by Michael (2006) Practical Lock-free and wait-free ..., at page 2. He has a reasoning why double CAS is not feasible in modern systems.}
Since the beginning of shared memory computing, programmers and researchers have had to deal with the ABA problem:
Even though a process retrieves the same value twice in a row from a shared memory object, it is still possible that the value of the object has changed multiple times.

Especially algorithms using the standard Compare-and-Swap (CAS) primitive seem to be susceptible.
A CAS object provides two operations: \Read{} returns the value of the object, and \CAS{$x$,$y$} changes the value of the object to $y$ provided that its value $v$ prior to the operation equals $x$, and it returns $v$.
(According to some specifications a \CAS{$x$,$y$} returns a Boolean, which is \True if and only if the \CAS{} succeeded, i.e., it wrote $y$.)
\oldpwtodo{Indicate in each section, which specification we use.}
Often, \CAS{} objects are used in the following way:
First, a process $p$ reads the value $x$ stored in the CAS object, then it performs some computation, and finally it tries to propagate the result of the computation by performing a \CAS{$x$,$y$}.
The idea is that if another process has already updated the data structure, $p$'s \CAS{} should fail, and so inconsistencies are avoided.
However, if multiple successful \CAS{} operations have occurred and the value of the object has changed back to $x$, $p$'s \CAS{} might still succeed, possibly yielding inconsistencies.

ABAs are also a problem for algorithms using other strong primitives, or even only registers.
For example, in mutual exclusion algorithms often processes busy-wait for certain events to happen, by repeatedly reading the same register.
In systems with caches, the cost of waiting is small, because as long as no process changes the register value, all reads are cache hits.
The event is signaled by other processes through a change in the register value.
But it may also be desirable to eventually reset the register to its state, before the event was signaled, in order to be able to reuse it.
But this may result in the ABA problem, and as a consequence waiting processes may miss events.
Therefore, algorithm designers have to devise more complicated code in order to avoid unnoticed cache misses, or even lack of progress.
\pwtodo{Remove paragraph above, if we need more space.}

Many shared memory algorithms and data structures have to deal with the ABA problem.
Often this is done in an ad-hoc, application specific way \cite{TZ2001a}, or solutions are based on tagging \cite{SHC2000a,Sto1990a,MS1996a,MS1998a,HSY2010a,PLJ1991a,Mic2002a,Mic2004a} (see below).
Other papers combine tagging and memory management techniques, or suggest both as alternatives \cite{LS2008a,HSY2010a}.

Tagging, introduced by IBM \cite{IBM1983a}, requires augmenting an object with a tag (which is sometimes called sequence number) that gets changed with every write operation.
This technique avoids the ABA problem only, if tags never repeat.
Therefore, theoretically, an infinite number of tags and thus base objects of unbounded size are required.
One may argue that, in practice, for reasonably large base objects, a system will never run out of tags.
However, this is unrealistic in cases where the tag has to be stored together with other information in the same object.
In some cases, it is possible to store the tag in a separate object (e.g., \cite{JP2003a}), but this requires technically difficult algorithms and tedious correctness proofs.
Some architectures like the IBM System/370 \cite{IBM1983a} introduced a double-width CAS primitive, which allows one of two (32-bit) words to be used for storing tags.
While using bounded tags does not completely avoid the ABA problem (because tag values may wrap around), it has been argued \cite{MS1996a,MS1998a,SHC2000a,Sto1990a} that an erroneous algorithm execution due to an unexpected ABA becomes very unlikely.
From a theoretical perspective this is unsatisfactory.
Moreover, for practical applications, it is often necessary to use the entire object space (today usually comprising 64 bits) for data, so the tagging technique requires double-width atomic instructions.
Those are not supported by most mainstream architectures \cite{Mic2004c}.

ABAs cause problems in algorithms that use some form of memory management, where a pointer to some memory space may change its value in an ABA fashion.
In this context, memory reclamation techniques based on reference counting \cite{Val1995a}, Hazard pointers~\cite{Mic2004b,Mic2004c}, the repeat-offender problem technique~\cite{HLM2002a}, or the memory reclamation technique introduced in \cite{AGW2014a} deal with the ABA problem.
But those techniques are application specific.

A more methodological approach has been followed by research that showed how a load-linked store-conditional (LL/SC) object can be implemented from CAS objects and registers.
Such an object provides two operations, \LL{} and \SC{}, where \LL{} 
returns the current value of the object.
\SC{$x$} may either fail and not change anything, or succeed and write the value $x$ to the object.
Specifically, an \SC{$x$} operation by process $p$ succeeds if and only if no other \SC{} operation succeeded since $p$'s last \LL{}.
A Boolean return value of an \SC{} operation indicates its success (\True) or failure (\False).
An extended specification also allows for a \VL{} (verify-link) operation, which does not change the state of the object, but it returns \False if a successful \SC{} has been performed since the calling process' last \LL{}, and \True otherwise.
LL/SC (or LL/SC/VL) objects can in almost all cases replace CAS objects in algorithms, and are an effective way of avoiding the ABA problem.
Unfortunately, existing multiprocessor systems only provide weak versions of LL/SC that restrict programmers severely in how they can use the objects \cite{Moi1997a}, and hence they ``offer little or no help with preventing the ABA problem'' \cite{Mic2004d}.

For that reason, a line of research has been dedicated to finding time and space efficient LL/SC implementations from CAS objects and registers \cite{AM1995a,Moi1997a,JP2003a,DHLM2004a,JP2006a,Mic2004c,Mic2004d}.
While many of those solutions are wait-free and often even guarantee constant time execution of each \LL{} and \SC{} operation, they still have drawbacks:
Existing implementations either require unbounded tags (e.g., \cite{Moi1997a}) and thus use unbounded CAS objects or registers, or they need at least linear space. 
Jayanti and Petrovic \cite{JP2003a} and Anderson and Moir \cite{AM1995a} presented the most space efficient implementations of an LL/SC object from bounded CAS and registers, which achieve constant step-complexity: they use only one CAS object but require $\Theta(n)$ registers.
This raises the question, whether time efficient implementations of LL/SC from a smaller number of bounded CAS objects and registers may exist.
More generally, in order to understand the power and limits of shared memory primitives, it seems important to learn how much time and space is required to avoid or detect ABAs, and not to restrict this question to the implementation of LL/SC objects from CAS objects and registers.

CAS and LL/SC objects have a consensus number of infinity \cite{Her1991a}, while registers have a consensus number of one.
Therefore, it is impossible to implement wait-free LL/SC from registers or other objects with a bounded consensus number.
Time and space lower bounds for implementations of LL/SC objects may not necessarily imply that it is the ABA problem that is hard to solve, but such lower bounds may follow inherently from other properties of the LL/SC specification.

\paragraph{Results.}
To investigate the complexity of detecting or preventing ABAs, we define a natural object, the \emph{ABA-detecting register}.
It supports two operations, \DRead{} and \DWrite{}.
Operation \DWrite{$x$} writes value $x$ to the register, and returns nothing.
Operation \DRead{} by process $p$ returns, in addition to the value of the register, a Boolean flag, which is \True if and only if some process executed a \DWrite{} since $p$'s last \DRead{} operation.
We distinguish between \emph{single-writer} ABA-detecting registers, where only one dedicated process is allowed to call \DWrite{}, and \emph{multi-writer} ones that don't have this restriction.

A wait-free ABA-detecting register can be implemented from registers, and thus has consensus number 1.
(Therefore, they are weaker with respect to wait-freedom than CAS or LL/SC.)
Using a single unbounded register with an unbounded tag that gets changed whenever some process writes to it, it is trivial to obtain an ABA-detecting register with constant time complexity.
But if base objects have only bounded size, the situation is completely different:
For implementations of ABA-detecting registers in a system with $n$ processes and bounded registers, we obtain a linear (in $n$) space lower bound, even if the implementation satisfies only nondeterministic solo-termination (the non-deterministic variant of obstruction-freedom), which is a progress condition strictly weaker than wait-freedom.
The availability of CAS seems to be of little help:
For wait-free implementations from CAS objects and registers we obtain a time-space tradeoff that is linear in $n$.
The same asymptotic time-space tradeoff is obtained if the base objects support arbitrary conditional read-modify-write operations~\cite{FHS2006a}.
Each conditional operation can be simulated by a single operation on a \emph{writable} CAS objects, i.e., an object that supports a \Write{} operation in addition to \Read{} and \CAS{}.
For that reason we state the lower bound for implementations from conditional read-modify-write operations in terms of writable CAS base objects. 
\begin{theorem}\label{thm:mainlb:ABA_detecting}
  Any linearizable implementation of a single-writer 1-bit ABA-detecting register from $m$ bounded base objects satisfies:
  \begin{enumerate}[label=(\alph*)]
    \item $m\geq n-1$ if the base objects are bounded registers, and the implementation satisfies nondeterministic solo-termination;
    \item $m\geq (n-1)/t$, if the the base objects are bounded CAS objects and registers, and the implementation is deterministic and wait-free with worst-case step-complexity at most $t$; and
    \item $m\geq (n-1)/(2t)$, if the base objects are bounded writeable CAS objects, and the implementation is deterministic and wait-free with worst-case step-complexity at most $t$.
  \end{enumerate}
\end{theorem}
The requirement that base objects are bounded is necessary for this lower bound, because, as mentioned earlier, an ABA-detecting register can be trivially obtained by augmenting a normal register with an unbounded tag.

There is a simple implementation of a (bounded) ABA-detecting registers with constant step-complexity from a single (bounded) LL/SC/VL object of the same size:
Each process uses a local variable $old$.
To \DWrite{$x$}, the process executes a \LL{} operation followed by a \SC{$x$}.
To \DRead{}, the process first executes a \VL{}. 
If \VL{} returns \True, the process returns $(old, \False)$; otherwise, it reads the value of the LL/SC/VL object into $old$ (by executing \LL{}), and then returns $(old, \True)$.
It is not hard to see that this implementation is linearizable.
(See Appendix~\ref{sec:LL/SC/VL->ABA-Detecting} for the algorithm and proof of correctness.)
Thus, by reduction we obtain the same lower bound as the one stated in \cref{thm:mainlb:ABA_detecting} for implementations of single bit LL/SC/VL.
Unfortunately, for that reduction the \VL{} operation is needed, and at least we do not know how to obtain a similarly efficient ABA-detecting register from an LL/SC object that does not support \VL{}.
However, the proofs of \cref{thm:mainlb:ABA_detecting} can be easily modified to accommodate LL/SC objects:
\begin{corollary}\label{cor:mainlb:LLSC}
  Any linearizable implementation of a single bit LL/SC object from $m$ bounded base objects satisfies
  \begin{enumerate}[label=(\alph*)]
    \item $m\geq (n-1)/t$, if the the base objects are bounded CAS objects and registers, and the implementation is deterministic and wait-free with worst-case step-complexity at most $t$; and
    \item $m\geq (n-1)/(2t)$, if the base objects are bounded writeable CAS objects, and the implementation is deterministic and wait-free with worst-case step-complexity at most $t$.
  \end{enumerate}
\end{corollary}
A linear space lower bound (corresponding to Part~(a) of \cref{thm:mainlb:ABA_detecting}) for nondeterministic solo-terminating implementations of LL/SC from (even unbounded) registers follows from the fact that LL/SC objects are \emph{perturbable} \cite{JTT2000a}.

As in \cref{thm:mainlb:ABA_detecting}, the assumption that base objects are bounded is necessary, because there is an implementation of an LL/SC/VL object from a single unbounded CAS object with constant step complexity by Moir \cite{Moi1997a}.
Our time-space tradeoff is asymptotically tight for implementations with constant step-complexity, as it matches known upper bounds \cite{AM1995a,JP2003a}.
We show that it also asymptotically tight for implementations using a single CAS object:
\begin{theorem}\label{Thm:LLSC-from-CAS}
  A single bounded CAS object suffices to implement a bounded LL/SC/VL object or a bounded multi-writer ABA-detecting register with $O(n)$ step-complexity.
\end{theorem}

These results raise the question, whether bounded CAS objects are helpful for ABA detection.
We determine that for this problem bounded CAS objects do not provide additional benefits over bounded registers:
\begin{theorem}\label{Thm:ABAReg-from-Reg}
  There is a linearizable wait-free implementation of a multi-writer $b$-bit ABA-detecting register from $n+1$ $(b+2log n+O(1))$-bit registers with constant step complexity.
\end{theorem}

Not only do our lower bounds show that Anderson's and Moir's \cite{AM1995a} as well as Jayanti's and Petrovic's \cite{JP2003a} implementations of LL/SC from CAS objects and registers are optimal with respect to their time- and space-product, but they also clearly indicate that ABA detection is inherently difficult, even if bounded conditional read-modify-write primitives such as (writable) CAS objects are available.
Therefore, other primitives that provide a solution to the ABA problem would most likely be as difficult to obtain as LL/SC.
Our upper bounds demonstrate that bounded CAS objects (and in fact any conditional read-modify-write operations) are not more helpful than bounded registers with respect to ABA detection.
On the other hand, ABA detection is difficult only if base objects are bounded, but for our lower bounds it does not matter how large that bound on the size of the base object is, as long as it is finite.

%
%
%

\paragraph{Other Related Work.}
Our lower bounds use covering arguments.
Covering arguments were first used by Burns and Lynch \cite{BL1993a} to prove a space lower bound for mutual exclusion, and essentially all space lower bounds are based on this technique.
Examples are space lower bounds for one-time test-and-set objects \cite{SP1989a}, consensus \cite{FHS1998a}, timestamps \cite{EFR2008a,HHPW2014a}, and the general class of perturbable objects \cite{JTT2000a} (which includes LL/SC among others).
These lower bounds have in common that they do not apply if CAS objects are available as base objects.
(They allow for registers, swap objects, and, in case of \cite{JTT2000a}, resettable consensus.)
An overview of covering arguments can be found in Attiya's and Ellen's recent textbook\cite{AE2014:impossibility_results}.

In our time-space tradeoffs we construct executions, where a sequence of operations by a process $p$ is interleaved with successful \CAS{} and \Write{} operations of other processes, so that $p$'s steps remain ``hidden''.
Such a technique has been also used by Fich, Hendler, and Shavit \cite{FHS2006a} to prove linear space lower bounds for wait-free implementations of \emph{visible} objects implemented from conditional read-modify-write (i.e., writable CAS) objects.
Visible objects include counters, queues, stacks, or snapshots.
Neither ABA-detecting registers nor LL/SC objects are visible, because they can be implemented from a single unbounded CAS object.
In fact, we are not aware of any other non-trivial lower bounds that, like ours, separate bounded from unbounded base objects.

\paragraph{Preliminaries.}
We consider a system with $n$ processes with unique IDs in $\PP=\set{0,\dots,n-1}$.
Processes communicate through shared memory operations, called steps, that are executed on atomic base objects provided by the system.
Each process executes a possibly nondeterministic program.
If processes are deterministic, a \emph{schedule} is a sequence of process IDs, that determines the order in which processes execute their steps.
If processes are nondeterministic, a \emph{schedule} is a sequence of process IDs together with coin-flips, and it describes the order in which processes take steps together with the nondeterministic decisions they make.
The sequence of shared memory steps taken by processes is called \emph{execution}.
A \emph{history} on some implemented object is the sequence of method call invocations and responses that occur in an execution on that object.
A \emph{configuration} describes the state of the system, i.e., of all processes and all base objects.

Our implementations are deterministic and \emph{wait-free}, which means that every method call terminates within a finite number of the calling process' steps, in any execution.
The step-complexity of a deterministic wait-free method is the maximum number of steps a process needs to terminate the method call in any execution.
Our lower bounds hold for implementations that satisfy a progress condition which is strictly weaker than wait-freedom:
A nondeterministic method $m$ satisfies \emph{nondeterministic solo-termination}, if for every process $p$ and every configuration $C$ in which a call of method $m$ by $p$ is pending, there is a $p$-only execution that starts in $C$ and during which $p$ finishes method $m$.
For deterministic algorithms, nondeterministic solo-termination is the same as \emph{obstruction-freedom}.
Our algorithms are \emph{linearizable} \cite{HW1990a}, but our lower bounds work for much weaker correctness conditions.

%% file: lowerbounds.tex
\section{Lower Bounds}

\oldpwtodo{Define contatenation operator for scheules/executions}
For a configuration $C$ and a schedule $\sigma$, let $\Exec(C,\sigma)$ denote the execution arising from processes taking steps, starting in configuration $C$, in the order defined by $\sigma$, and using the nondeterministic decisions defined by $\sigma$, if the algorithm is nondeterministic.
Let $\Conf(C,\sigma)$ denote the configuration resulting from execution $\Exec(C,\sigma)$ started in $C$.
For two configurations $C$ and $D$ and a schedule $\alpha$, we write $C\yields{\alpha} D$ to indicate that $\Conf(C,\alpha)=D$. 
Let $C_{init}$ denote the initial configuration.
If there exists a schedule $\alpha$ such that $C\yields{\alpha} D$, then we say $D$ is reachable from $C$, and if $D$ is reachable from $C_{init}$, we simply say $D$ is reachable.

An execution $E$ or a schedule $\alpha$ is $P$-only for a set $P\subseteq\set{0,\dots,n-1}$ of processes, if only processes in $P$ take steps during $E$ respectively $\alpha$.
If $P=\set{p}$ is the set of a single process, then we sometimes write $p$-only instead of $\set{p}$-only.

For an execution $E$, let $\prec_E$ denote the happens-before order on operations in $E$, i.e., if operation $op$ responds in $E$ before $op'$ gets invoked, then and only then $op\prec_E op'$ ($op$ happens before $op'$).
We write simply $\prec$ instead of $\prec_E$, if is clear from the context which execution $E$ the relation refers to.
For a schedule $\alpha$, an execution $E$ and a process $p$, $E|p$ and $\alpha|p$ denote the sub-sequences of steps by $p$ in $E$ and in $\alpha$, respectively.

Two configurations $C$ and $D$ are \emph{indistinguishable} to process $p$, if every register has the same value in $C$ as in $D$, and $p$ is in the same state in both configurations.
We write $C\indist_p D$ to denote that $C$ and $D$ are indistinguishable to $p$.
We write $C\indist_S D$ for a set $S$ of processes to denote that $C\indist_p D$ for every process $p\in S$.
We say process $p$ is \emph{idle} in configuration $C$, if it has no pending method call, and if all processes are idle, then the configuration is \emph{quiescent}.
A process \emph{completes} a method call in an execution $E$, if that method terminates in $E$.

For our lower bounds, we do not require that the implementation of the ABA-detecting registers is linearizable.
Instead, we consider methods \WRead{} and \WWrite{} that take no arguments, and where \WRead{} returns a Boolean value, and \WWrite{} returns nothing.
A correct concurrent implementation of these methods must guarantee for every execution, that a \WRead{} operation $r$ by process $p$ returns \True if and only if there exists a \WWrite{} operation $w$ such that $w$ happens before $r$ and every other \WRead{} operation by $p$ happens before $w$. 

Linearizability of an ABA-detecting register $R$ guarantees that the operations $R.\DRead{}$ (in place of \WRead{}) and $R.\DWrite{}$ (in place of \WWrite{}) satisfy the correctness properties above.
Therefore, every lower bound on the time and/or space complexity for correct implementations of those methods implies the same lower bound for linearizable ABA-detecting registers.

Let $p$ be some process and $C$ a configuration.
We say $C$ is \emph{$p$-clean}, if there exists a schedule $\alpha$, $C_{init}\yields{\alpha} C$, such that $\Exec(C_{init},\alpha)$ contains a complete $\WRead{}$ operation $r^\ast$ by $p$, and every \WWrite{} happens before $r^\ast$.
Configuration $C$ is \emph{$p$-dirty}, if there exists a schedule $\alpha$, $C_{init}\yields{\alpha} C$, and $\Exec(C_{init},\alpha)$ contains a complete \WWrite{} operation $w^\ast$ such that no $\WRead{}$ by $p$ is pending at any point after $w^\ast$ has been invoked.
Note that some configurations are neither $p$-dirty nor $p$-clean.

Throughout this section we assume that each process executes an infinite program, in which it repeatedly calls \WRead{} and \WWrite{} methods.
More specifically, process 0 repeatedly executes \WWrite{}, while every process in $\set{1,\dots,n-1}$ repeatedly calls \WRead{}.

Then in a $p$-only execution starting from a configuration $C$, the first \WRead{} operation by $p$ returns \False if $C$ is $p$-clean and \True if $C$ is $p$-dirty. 
Therefore, each process must be able to distinguish $p$-clean configurations from $p$-dirty ones.
The full proof of the following observation can be found in \cref{sec:proof:obs:clean-dirty}.
\begin{observation}\label{obs:clean-dirty-distinguishability}
  Suppose the $\WRead{}$ method satisfies nondeterministic solo-termination.
  For any process $p\in\set{1,\dots,n-1}$ and any two reachable configurations $C_1,C_2$, if $C_1$ is $p$-clean and $C_2$ is $p$-dirty, then $C_1\nindist_p C_2$.
\end{observation}

\subsection{A Space Lower Bound for Implementations from Bounded Registers}
Let $\RR$ be a set of $k$ registers and $P$ a set of processes.
We say the processes in $P$ \emph{cover} $\RR$ in configuration $C$, if for each register $R\in\RR$ there is a process in $P$ that is poised to write to $R$.
A \emph{block-write} to $\RR$ is an execution in which $k$ processes participate, and each of them takes exactly one step in which it writes to a distinct register in $\RR$.
(The only block-write to $\emptyset$ is the empty execution.)

In the following we assume an implementation of methods $\WRead{}$ and $\WWrite{}$ from $m$ bounded registers.
The \emph{register configuration} of a configuration $C$ is an $m$-tuple, $\reg(C)=(v_1,\dots,v_m)$, where $v_i$ is the value of the $i$-th register.

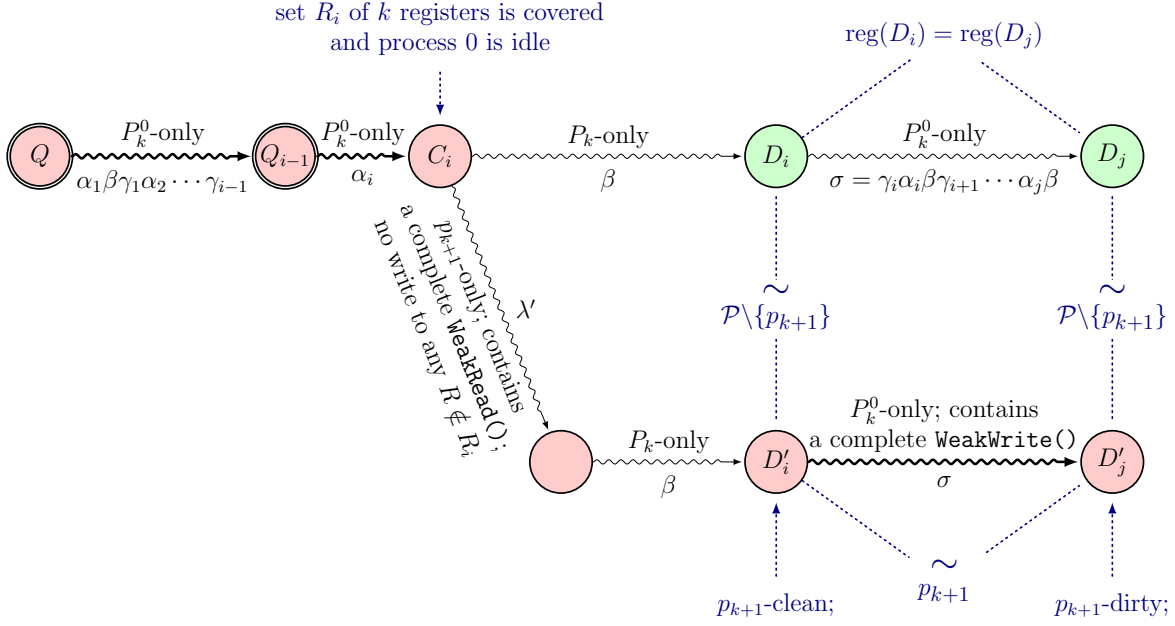
\begin{figure*}[!htb]
	\begin{center}\large
		\resizebox{.95\textwidth}{!}{
			\begin{tikzpicture}[auto, thick, line cap=round, node distance=2cm]
			\node [quiescent] (Q) {$Q$};
			\node [quiescent,right=3 of Q] (Qi-1) {$Q_{i-1}$};
			\node [config,right=1.5 of Qi-1] (Ci) {$C_i$};
			\node [config,right=4.5 of Ci,fill=green!20] (Di) {$D_i$};
			\node [config,right=4.5 of Di,fill=green!20] (Dj) {$D_j$};
			\node [config,below=4 of Di] (Di') {$D_i'$};
			\node [config,left=2.5 of Di'] (Ci') {};			
			\node [config,below=4 of Dj] (Dj') {$D_j'$};
			
			\path [yields,very thick] (Q) -- node [schedulelabel] {$\alpha_1\beta\gamma_1\alpha_2\cdots\gamma_{i-1}$} node [scheduledescr,black] {$P_k^0$-only} (Qi-1);
			\path [yields,very thick] (Qi-1) -- node [schedulelabel] {$\alpha_i$} node [scheduledescr] {$P_k^0$-only} (Ci);
			\path [yields] (Ci) -- node [schedulelabel] {$\beta$} node [scheduledescr] {$P_k$-only} (Di);
			\path [yields] (Di) -- node [schedulelabel] {$\sigma=\gamma_i\alpha_i\beta\gamma_{i+1}\cdots\alpha_j\beta$} node [scheduledescr] {$P_k^0$-only} (Dj);
			\path [yields] (Ci) -- node [schedulelabel,right=3pt] {$\lambda'$} node [scheduledescr,below,sloped] {$p_{k+1}$-only; contains \\a complete \WRead{};\\no write to any $R\notin R_i$} (Ci');
			\path [yields] (Ci') -- node [schedulelabel] {$\beta$} node [scheduledescr] {$P_k$-only} (Di');			
			\path [yields,very thick] (Di') -- node [schedulelabel] {$\sigma$} node [scheduledescr] {$P_k^0$-only; contains \\a complete \WWrite{}} (Dj');
			
			\node [indistlabel,above=1 of Ci,align=center] (covered) {set $R_i$ of $k$ registers is covered\\ and process 0 is idle};
			\path [indist,arrow] (covered) -- (Ci);
			
			\path [indist] (Di) to node [indistlabel] {\LARGE$\displaystyle\indist_{\PP\setminus\set{p_{k+1}}}$} (Di');

			\path [indist] (Dj) to node [indistlabel] {\LARGE$\displaystyle\indist_{\PP\setminus\set{p_{k+1}}}$} (Dj');
			\path (Di') -- (Dj') node [midway,yshift=-2cm] (Di'Dj') [indistlabel] {\LARGE$\displaystyle\indist_{p_{k+1}}$};
			\path [indist] (Di') -- (Di'Dj') -- (Dj');
			
			\path (Di) -- (Dj) node [midway,yshift=2cm] (DiDj) [indistlabel] {$\reg(D_i)=\reg(D_j)$};
			\path [indist] (Di) -- (DiDj) -- (Dj);
			
			\node [indistlabel,below=1.5 of Di',align=center] (clean) {$p_{k+1}$-clean;};
			\path [indist,arrow] (clean) -- (Di');

			\node [indistlabel,below=1.5 of Dj',align=center] (dirty) {$p_{k+1}$-dirty;};
			\path [indist,arrow] (dirty) -- (Dj');
			\end{tikzpicture}}
	\end{center}
	\caption{Proof of \cref{lem:registerlb_main}. Let $P_k^0$ denote the set $P_k\cup\set{0}$. Double circles denote quiescent configurations.
		\label{fig:lem:reglb_main}}
\end{figure*}

\begin{lemma}\label{lem:registerlb_main}
  Suppose methods $\WRead{}$ and $\WWrite{}$ satisfy nondeterministic solo-termination. 
  For any quiescent configuration $Q$ and any set $P_k=\set{p_1,\dots,p_k}\subseteq\PP\setminus\set{0}$, where $k\in\set{0,\dots,n-1}$, there exists a $(P_k\cup\set{0})$-only schedule $\alpha$ such that in $\Conf(Q,\alpha)$ process 0 is idle and $k$ distinct registers are covered by $p_1,\dots,p_k$.
\end{lemma}
\cref{lem:registerlb_main} immediately implies \cref{thm:mainlb:ABA_detecting}\,(a).
\pwtodo{Journal version: add a picture}
\begin{proof}[Proof of \cref{lem:registerlb_main}]
	The proof is by induction on $k$.
	If $k=0$, we let $\alpha$ be the empty schedule, and the lemma is immediate because $\Conf(Q,\alpha)=Q$ is a quiescent configuration (so 0 is idle).
	
	Now suppose we have proved the inductive hypothesis for some integer $k<n-1$.
	Let $\beta=(p_1,\dots,p_k)$ be the schedule in which each of $p_1,\dots,p_k$ takes exactly one step.
	Let $Q_0=Q$.
	By the inductive hypothesis there is a schedule $\alpha_1$ such that in $C_1:=\Conf(Q_0,\alpha_1)$ a set $\RR_1$ of exactly $k$ registers is covered, and process 0 is idle.
	Hence, $\Exec(C_1,\beta)$ is a block-write to $\RR_1$ yielding a configuration $D_1=\Conf(C_1,\beta)$.
	We let $\gamma_1$ be the schedule such that in $\Exec(D_1,\gamma_1)$ first each process in $\set{p_1,\dots,p_k}$ takes enough unobstructed steps to finish its \WRead{} method call, and after that process 0 takes enough unobstructed steps to complete exactly one \WWrite{} method.
	Then $Q_1=\Conf(D_1,\gamma_1)$ is quiescent, and during $\Exec(D_1,\gamma_1)$ exactly one complete \WWrite{} gets executed.  
	Repeating this construction (using the inductive hypothesis repeatedly) we obtain a schedule
	\begin{math}
		\alpha_1\beta\gamma_1\alpha_2\beta\gamma_2\alpha_3\dots
	\end{math}
	and configurations $Q_0,C_1,D_1,Q_1,C_2,D_2,Q_2,\dots$ and sets of $k$ registers $\RR_1,\RR_2,\dots$, such that for any $i\geq 1$:
	\begin{itemize}
		\item $Q_{i-1}\yields{\alpha_i} C_i\yields{\beta} D_i\yields{\gamma_i} Q_i$;
		\item $Q_i$ is quiescent;
		\item during $\Exec(D_i,\gamma_i)$ process 0 executes a complete \WWrite{} operation; and
		\item in $C_i$ process 0 is idle and $\RR_i$ is covered by $P_k$ (and thus $\Exec(C_i,\beta)$ is a block-write to $\RR_i$).
	\end{itemize}
	Since the number of registers is finite, and all registers are bounded, there exist indices $1\leq i<j$ such that 
	$\reg(D_i)=\reg(D_j)$.
	Let $\sigma=\gamma_i\alpha_{i+1}\beta\gamma_{i+1}\alpha_{i+2}\dots\alpha_j\beta$, i.e.,
	\begin{equation}\label{eq:register_lb:25}
	C_i\yields{\beta} D_i\yields{\sigma} D_j.
	\end{equation}
	This situation is depicted in \cref{fig:lem:reglb_main}. 
	Now let $\lambda'$ be a  $p_{k+1}$-only schedule such that in $\Exec(C_i,\lambda')$ process $p_{k+1}$ completes exactly one $\WRead{}$ method call.
	By the nondeterministic solo-termination property, such a schedule $\lambda'$ exists.
	Let $\lambda$ be the prefix of $\lambda'$, such that $\Exec(C_i,\lambda)$ ends when $p_{k+1}$ is poised to write to a register $R\not\in\RR_i$ for the first time, or $\lambda=\lambda'$ if $p_{k+1}$ finishes its \WRead{} method call without writing to a register outside of $\RR_i$.
	
	First assume $\lambda\neq\lambda'$, i.e., in $\Exec(C_i,\lambda)$ process $p_{k+1}$ does not finish its \WRead{} method call, but instead the execution ends when $p_{k+1}$ covers a register $R\not\in\RR_i$.
	Since in $C_i$ process $0$ is idle and $\RR_i$ is covered by $P_k$, and since $\lambda$ is $p_{k+1}$-only, in configuration $\Conf(C_i,\lambda)=\Conf(Q,\alpha_1\beta\gamma_1\dots\alpha_i\lambda)$ processes $p_1,\dots,p_{k+1}$ cover $k+1$ registers, and process 0 is still idle.
	This completes the proof of the inductive step for $\alpha=\alpha_1\beta\gamma_1\dots\alpha_i\lambda$.
	
	Now we consider the case $\lambda=\lambda'$, i.e., during $\Exec(C_i,\lambda)$ process $p_{k+1}$ finishes its \WRead{} method call without writing to a register outside of $\RR_i$.
	To complete the proof of the lemma, it suffices to show that this case cannot occur.
	(This case is depicted in \cref{fig:lem:reglb_main}.)
	
	Since in $C_i$ the processes in $P_k$ cover $\RR_i$, and $p_{k+1}$ only writes to registers in $\RR_i$ during $\Exec(C_i,\lambda)$, it follows that $\Exec(C_i,\lambda\beta)$ ends with a block-write by $P_k$ in which all writes by $p_{k+1}$ get obliterated.
	In particular, for $D_i':=\Conf(C_i,\lambda\beta)$ we have
	\begin{equation}
	D_i'\indist_{\PP\setminus\set{p_{k+1}}} D_i\label{eq:register_lb:50}
	\end{equation}
	Hence, since schedule $\sigma$ is $(P_k\cup\set{0})$-only, i.e., $p_{k+1}$ does not participate, we obtain
	\begin{math}
		\Exec(D_i',\sigma)=\Exec(D_i,\sigma),
	\end{math}
	and in particular using \cref{eq:register_lb:25} 
	\begin{equation}\label{eq:register_lb:75}
	C_i\yields{\lambda\beta}D_i'\yields{\sigma}D_j',
	\quad\text{where}\quad
	D_j\indist_{\PP\setminus\set{p_{k+1}}} D_j'.
	\end{equation}
	Now recall that we chose $i$ and $j$ in such a way that $\reg(D_i)=\reg(D_j)$.
	Thus, from \cref{eq:register_lb:50,eq:register_lb:75} we get 
	\begin{math}
		\reg(D_i')=\reg(D_i)=\reg(D_j)=\reg(D_j').
	\end{math}
	Because $D_i'\yields{\sigma} D_j'$ (\cref{eq:register_lb:75}), and since by construction only processes $\set{0,p_1,\dots,p_k}$ appear in $\sigma$, $p_{k+1}$ is in $D_i'$ in exactly the same state as in $D_j'$.
	Hence,
	\begin{equation}
	D_i'\indist_{p_{k+1}} D_j'.\label{eq:register_lb:150}
	\end{equation}
	Now recall that $C_i\yields{\lambda\beta} D_i'$, and in the corresponding execution process $p_{k+1}$ executes a complete \WRead{} method, while process 0 takes no steps, and $p_{k+1}$ is idle in $D_i'$.
	Hence, $D_i'$ is $p_{k+1}$-clean.
	On the other hand, $\Exec(D_i',\sigma)=\Exec(D_i,\sigma)$ starts with a complete \WWrite{} operation (during the prefix $\Exec(D_i',\gamma_i)$) by process 0, while process $p_{k+1}$ takes no steps, and thus remains idle.
	It follows that the configuration resulting from that execution, $D_j'$, is $p_{k+1}$-dirty.
	Summarizing, we have two reachable configurations, $D_i'$ and $D_j'$, where one of them is $p_{k+1}$-clean and the other one is $p_{k+1}$-dirty, and both are indistinguishable to $p_{k+1}$, according to \cref{eq:register_lb:150}.
	This contradicts \cref{obs:clean-dirty-distinguishability}.
\end{proof}

\subsection{A Time-Space Tradeoff for Implementations from CAS Objects}
We now consider deterministic wait-free implementations of \WRead{} and \WWrite{} from $m$ writable bounded CAS objects. 
We assume without loss of generality that every \CAS{$x$,$y$} operation satisfies $x\neq y$. (A \CAS{$x$,$x$} operation can be replaced by a \Read{}).

For any configuration $C$ and any shared CAS object $R$ let $\CCov(C,R)$ and $\WCov(C,R)$ denote the sets of processes that are poised in $C$ to execute a \CAS{} respectively \Write{} operation on $R$.
Let $P$ be a set of processes and $C$ a configuration.
A schedule $\beta$ is called \emph{$P$-successful} for $C$, if it contains every process in $P$ exactly once, and every step of $\Exec(C,\beta)$ is either a \Write{} or a successful \CAS{}.
If a configuration $C$ has a $P$-successful schedule $\beta$, then we also say 
execution $\Exec(C,\beta)$ is $P$-successful.

As before, we assume that all processes run an infinite loop, where process 0 repeatedly calls \WWrite{} while all other processes repeatedly call \WRead{}.

\begin{lemma}\label{lem:tst:hide_process}
  Let $P\subsetneq\PP\setminus\{0\}$, $q\in\PP\setminus P$, $q\neq 0$.
  Let $C$ be a configuration, in which either $q$ is idle, or in no execution starting from $C$ process $q$ executes more than $t$ shared memory steps before finishing a pending \WRead{} call.
  If $\beta$ is a $P$-successful schedule for $C$, then there is a schedule $\sigma$ such that
  \begin{equation}\label{eq:tst:claim}
    \Conf(C,\beta)\indist_{\PP\setminus\set{q}} \Conf(C,\sigma),
  \end{equation}
  and at least one of the following is the case:
  \begin{enumerate}[label=(\alph*)]
    \item In $\Conf(C,\sigma)$ process $q$ is idle;
    \item in $\Conf(C,\sigma)$ process $q$ is poised to write to some object $R$ and $|\WCov(C,R)\cap P|<t$; or
    \item in $\Conf(C,\sigma)$ process $q$ is poised to execute a \CAS{$x$,$y$} operation on some object $R$, where $x$ is the value of $R$ in configuration $\Conf(C,\sigma)$, and $|\WCov(C,R)\cap P|+|\CCov(C,R)\cap P|<t$.
  \end{enumerate}
\end{lemma}
\begin{proof}
  We prove the lemma by induction on $t$.
  If $t=0$, then $q$ is idle in $C$.
  Hence, for $\sigma=\beta$ we obtain \cref{eq:tst:claim} and Case~(a).
  
  Let $op_q$ be the step $q$ is poised to execute in $C$, and let $V$ be the object affected by $op_q$.  
  Further, let $\val_C(V)$ denote the value of $V$ in configuration $C$.  

  \textbf{Case~1:}
  First, assume that $op_q$ is a \Read{} or a \CAS{} operation.
  Let $z$ be the first process in $P$ that executes a step in $\Exec(C,\beta)|V$, and let $op_z$ be that step.
  We construct a two-step schedule $\lambda$ that contains $q$ and $z$, such that
  \begin{equation}\label{eq:tst:10}
      C':=\Conf(C,\lambda)\indist_{\PP\setminus\set{q}} \Conf(C,z).
  \end{equation}
  First suppose $op_z$ is a \CAS{$a,b$} operation and $op_q$ a \CAS{$x,y$} operation that would succeed in $C$ (i.e., $x=\val_C(V)$).
  Then we define $\lambda=(z,q)$.
  Since $\beta$ is $P$-successful, the \CAS{$a,b$} by $z$ in configuration $C$ succeeds and changes the value of $V$ from $a$ to $b$.
  In this case, $x=a$, so in the execution $\Exec(C,\lambda)$ the \CAS{$x,y$} by $q$ fails.
  \Cref{eq:tst:10} follows.
  
  In all other cases (i.e., if $op_z$ is a \Write{} operation or $op_q$ is a \Read{} or a \CAS{} that fails in $C$), then we let $\lambda=(q,z)$.
  Then in $\Exec(C,\lambda)$ either operation $op_q$ does not change the value of object $V$, or $op_z$ executes a \Write{} and overwrites any changes that may have resulted from $op_q$.
  It follows that \cref{eq:tst:10} is true.

  Now let $\beta'=\beta|(P\setminus\set{z})$, and recall that $C'=\Conf(C,\lambda)$.
  Since $\beta$ is $P$-successful 
  and in $\Exec(C,\beta)$ process $z$ executes the first step on $V$, 
  it follows that $\beta'$ is $P$-successful in $C'$.  
%

  Hence, we can apply the inductive hypothesis for configuration $C'$, process set $P'=P\setminus\{z\}$, and schedule $\beta'$, to obtain a schedule $\sigma'$ that satisfies one of the Cases~(a)-(c).
  Let $\sigma=\lambda\circ\sigma'$.
  Then by construction, 
  \begin{math}
    \Conf(C,\sigma)=\Conf(C',\sigma')\indist_{P\setminus\{q,z\}}\Conf(C',\beta')=\Conf(C,\beta)
  \end{math}.
  Because of \cref{eq:tst:10}, process $z$ can also not distinguish between $\Conf(C,\sigma)$ and $\Conf(C,\beta)$, so we obtain \cref{eq:tst:claim}.
  If (a) of the inductive hypothesis applies for $C'$ and $\sigma'$, then the same also applies for $C$ and $\sigma$, because $\Conf(C,\sigma)=\Conf(C',\sigma')$.
  Now suppose that Case~(b) applies for $C'$ and $\sigma'$.
  Let $R$ be the object on which process $q$ is poised to execute a \Write{} 
  in $\Conf(C',\sigma')=\Conf(C,\sigma)$.
  Starting from configuration $C'$, process $q$ must finish its \WRead{} method within $t'=t-1$ steps. 
  Hence, $|\WCov(C',R)\cap P'|\leq t'$.
  Since all processes other than $z$ are poised to execute the same step in $C$ as in $C'$, we have 
  $|\WCov(C,R)\cap P|\leq |\WCov(C',R)\cap P'|+1\leq t'+1=t$.
  Hence, Case~(b) follows for $C$, $\sigma$ and $P$.
  With exactly the same argument, if Case~(c) applies for $C'$, $\sigma'$, and $P'$, then it also applies for $C$, $\sigma$, and $P$.

  \textbf{Case~2:}
  We now assume that in $C$ process $q$ is poised to execute a \Write{} operation $op_q$ on object $V$.
  If $|\WCov(C,V)\cap P|<t$, we let $\sigma=\beta$.
  Then \cref{eq:tst:claim} and Case~(b) (for $R=V$) of the lemma are trivially satisfied.
  
  Hence, assume that $|\WCov(C,V)\cap P|\geq t$.
  Then $\Exec(C,\beta)$ contains at least $t$ writes to $V$.
  Let $z_1,\dots,z_{\ell-1}$ be the processes accessing $V$ in this order in $\Exec(C,\beta)|V$ before the first write to $V$ occurs, and let $z_\ell$ be the first process writing to $V$.
  Let $Z=\set{z_1,\dots,z_\ell}$, $\lambda=(z_1,\dots,z_{\ell-1},q,z_\ell)$, $\gamma=\beta|Z=(z_1,\dots,z_\ell)$, $\beta'=\beta|(P\setminus Z)$ and $P'=P\setminus Z$.
  Then in $\Exec(C,\gamma\circ\beta')$ all processes in $P$ execute exactly one step, as they do in $\Exec(C,\beta)$, and for each object $U$ the steps executed on $U$ occur in the same order in both executions.
  Hence, processes cannot distinguish these executions from each other, and in particular
  \begin{equation}\label{eq:tst:30}
  \Conf(C,\gamma\circ\beta')=\Conf(C,\beta).
  \end{equation}
  In $\Exec(C,\lambda)$, first processes $z_1,\dots,z_{\ell-1}$ execute successful \CAS operations on $V$, then $q$ writes to $V$, and finally, $z_\ell$ overwrites what $q$ has written.
  It follows that
  \begin{equation}\label{eq:tst:40}
    C':=\Conf(C,\lambda)\indist_{\PP\setminus\set{q}}\Conf(C,\gamma).
  \end{equation}
  Combining this with \cref{eq:tst:30} we obtain that $\beta'$ is $P'$-successful in $C'$.
  Moreover, since $q$ executed one step in the execution leading from $C$ to $C'$, in any execution starting from $C'$ it finishes its \WRead{} method after at most $t'=t-1$ steps.
  Thus, we can apply the inductive hypothesis to obtain a schedule $\sigma'$ such that $\Conf(C',\beta')\indist_{\PP\setminus\set{q}}\Conf(C',\sigma')$, and one of Cases (a)-(c) holds.
  Let $\sigma=\lambda\circ\sigma'$.
  Then $\Conf(C,\sigma)
  =\Conf(C',\sigma')\indist_{\PP\setminus\set{q}}\Conf(C',\beta')\indist_{\PP\setminus\set{q}}\Conf(C,\beta)$, where the last relation follows from \cref{eq:tst:30,eq:tst:40}.
  This proves \cref{eq:tst:claim}.
  
  If Case~(a) of the inductive hypothesis holds for $C'$ and $\sigma'$, then it is also true for $C$ and $\sigma$ because $\Conf(C,\sigma)=\Conf(C',\sigma')$.
  Now suppose either Case (b) or (c) applies to $C'$ and $\sigma'$.
  Let $R$ be the object process $q$ is poised to access in $\Conf(C,\sigma)=\Conf(C',\sigma')$.
  If $R\neq V$, then by construction we have $|\WCov(C,R)\cap P|=|\WCov(C',R)\cap P'|$ and $|\CCov(C,R)\cap P|=|\CCov(C',R)\cap P'|$.
  Hence, Case (b) or (c) for $C'$, $\sigma'$, and $P'$ immediately implies the same case for $C$, $\sigma$, an $P$.
  Finally, suppose Case~(b) or (c) occurs for $R=V$. 
  By construction in $\Exec(C,\lambda)$ only one process among all processes in $\WCov(C,V)$ writes to $V$, namely process $z_\ell$.
  Hence, in the configuration $C'$ reached by $\Exec(C,\lambda)$ all other processes in $\WCov(C,V)$ are still poised to write to $V$.
  Thus, for $R=V$ we obtain
  \begin{math}
  |\WCov(C',R)\cap P'|=|\WCov(C,R)\cap P|-1\geq t-1=t'
  \end{math},
  so neither Case~(b) nor Case~(c) can apply to $C'$, $\sigma'$, and $P'$---contradiction.
\end{proof}

\begin{figure*}[!htb]
	\begin{center}\large
		\resizebox{.95\textwidth}{!}{
			\begin{tikzpicture}[auto, thick, line cap=round, node distance=2cm]
			\node [quiescent] (Q) {$Q$};
			\node [quiescent,right=2.5 of Q] (Qi-1) {$Q_{i-1}$};
			\node [config,right=1.5 of Qi-1,fill=green!20] (Ci) {$C_i$};
			\node [config,right=3.5 of Ci] (Di) {$D_i$};
			\node [quiescent,right=3 of Di] (Qj-1) {$Q_{j-1}$};
			\node [config,right=of Qj-1,fill=green!20] (Cj) {$C_j$};
			\node [config,below=4 of Di] (Di') {$D_i'$};
			\node [config,below=4 of Cj] (Cj') {$C_j'$};  
			\node [config,right=2.5 of Cj'] (Dj') {$D_j'$};
			
			\path [yields,very thick] (Q) -- node [schedulelabel] {$\lambda$} node [scheduledescr,black] {$P_k^0$-only} (Qi-1);
			\path [yields,very thick] (Qi-1) -- node [schedulelabel] {$\alpha_i$} node [scheduledescr] {$P_k^0$-only} (Ci);
			\path [yields] (Ci) -- node [schedulelabel] {$\beta_i$} node [scheduledescr] {$P_k$-successful} (Di);
			\path [yields] (Di) -- node [schedulelabel] {$\lambda'$} node [scheduledescr] {$P_k^0$-only} (Qj-1);
			\path [yields] (Qj-1) -- node [schedulelabel] {$\alpha_j$} node [scheduledescr] {$P_k^0$-only} (Cj);
			\path [yields,very thick] (Ci) -- node [schedulelabel,above right=1pt] {$\sigma_i$} node [scheduledescr,below,sloped] {$P_{k+1}$-only;\\ Obtained by \cref{lem:tst:hide_process}} (Di');
			\path [yields,very thick] (Di') -- node [schedulelabel] {$\lambda'\alpha_j$} node [scheduledescr] {$P_k^0$-only;\\ contains a complete \Write{}} (Cj');
			\path [yields] (Cj') -- node [schedulelabel] {$\beta_i$} node [scheduledescr] {$P_k$-successful} (Dj');
			
			\path [indist] (Di) to node [indistlabel] {\LARGE$\displaystyle\indist_{\PP\setminus\set{p_{k+1}}}$} (Di');
			\path [indist] (Cj) to node [indistlabel] {\LARGE$\displaystyle\indist_{\PP\setminus\set{p_{k+1}}}$} (Cj');
			
			\path (Di') -- (Dj') node [midway,yshift=-2cm] (Di'Dj') [indistlabel] {\LARGE$\displaystyle\indist_{p_{k+1}}$};
			\path [indist] (Di') -- (Di'Dj') -- (Dj');
			
			\path (Ci) -- (Cj) node [midway,yshift=3cm] (CiCj) [indistlabel] {$\sig(C_i)=\sig(C_j)$};
			\path [indist] (Ci) -- (CiCj) -- (Cj);
			
			\node [indistlabel,below left=1 of Di',align=center] (clean) {if $p_{k+1}$ is idle, then $D_i'$ is $p_{k+1}$-clean;};
			\path [indist,arrow] (clean) -- (Di');
			
			\node [indistlabel,above=1.5 of Dj',align=center] (dirty) {if $p_{k+1}$ is idle,\\ then $D_j'$ is $p_{k+1}$-dirty;};
			\path [indist,arrow] (dirty) -- (Dj');
			\end{tikzpicture}}
	\end{center}
	\caption{Proof of \cref{lem:caslb_main}. Double circles denote quiescent configurations.
		\label{fig:lem:caslb_main}}
\end{figure*}
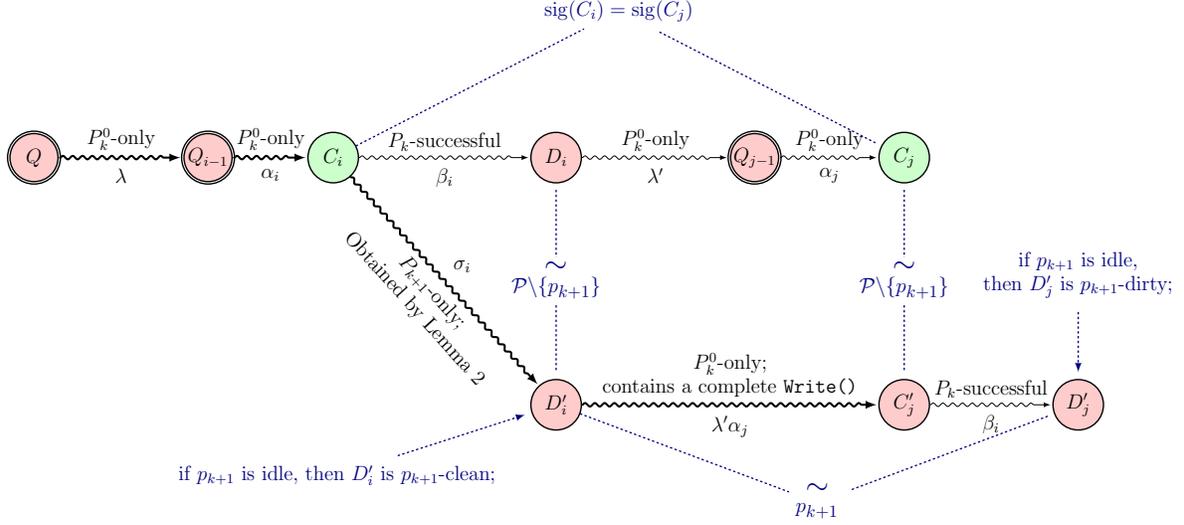

\begin{lemma}\label{lem:caslb_main}\samepage
Suppose \WRead{} and \WWrite{} have step complexity at most $t$. 
For any reachable quiescent configuration $Q$ and any set $P_k=\set{p_1,\dots,p_k}\subseteq\PP\setminus\set{0}$, where $k\in\set{0,\dots,n-1}$, there exists a $(P_k\cup\set{0})$-only schedule $\alpha$ such that $C:=\Conf(Q,\alpha)$ satisfies all of the following:
\begin{enumerate}[label=(\roman*)]
  \item all processes in $\PP\setminus P_k$ are idle in $C$;
  \item there is a $P_k$-successful schedule for $C$; and
  \item $|\WCov(C,R)\cap P_k|\leq t$ and $|\CCov(C,R)\cap P_k|\leq t$ for all objects $R$.
\end{enumerate}
\end{lemma}
\begin{proof}
  Throughout this proof let $P_k^0$ denote the set $P_k\cup\set{0}$.
  We prove the lemma by induction on $k$.
  For $k=0$, we let $\alpha$ be the empty schedule, so $C=\Conf(Q,\alpha)=Q$.
  Then $C$ is quiescent and (i) is true.
  Statements (ii)--(iii) follow immediately from $P_k=\emptyset$.
  
  Now suppose the inductive hypothesis is true for some value of $k\in\set{0,\dots,n-2}$.
  We let $Q_0=Q$ and $P_{k+1}=P_k\cup\set{p_{k+1}}$ for an arbitrary process $p_{k+1}\in\PP\setminus P_k^0$.
  Then, for $i=1,2,\dots$ we iteratively construct executions $\alpha_i,\beta_i,\gamma_i$ and configurations $Q_i,C_i,D_i$, where
  \begin{math}
  Q_{i-1}\yields{\alpha_i} C_i\yields{\beta_i} D_i\yields{\gamma_i} Q_i,
  \end{math}
  and $\alpha_i,\beta_i,\gamma_i$ are determined as follows:
  $\alpha_i$ is a $P_k^0$-only schedule that guarantees properties (i)-(iii) from the inductive hypothesis for configuration $Q_{i-1}$;
  $\beta_i$ is a $P_k$-successful schedule for $C_i$; and 
  $\gamma_i$ is an arbitrary $P_k$-only schedule followed by a 0-only schedule such that $Q_i$ is quiescent, and where $\Exec(D_i,\gamma_i)$ contains exactly one complete \WWrite{} operation by process 0.
  By the assumption that \WRead{} and \WWrite{} are wait-free, $\gamma_i$ exists.%
  
  We define for each configuration $C_i$ a signature, $\sig(C_i)$,
  which encodes for every process $p$ the exact shared memory operation $p$ is poised to execute next (including its parameters), and for every base object $R$ its value.
  
  
  Since there is only a finite number of bounded base objects in the system, there is a finite number of signatures, and thus there exist $1\leq i<j$ 
  such that $C_i$ and $C_j$ have the same signature. 
  We let 
  \begin{math}
  \lambda=\alpha_1\beta_1\gamma_1\alpha_2\dots\alpha_{i-1}\beta_{i-1}\gamma_{i-1}, 
  \quad
  \text{and}
  \quad
  \lambda'=\gamma_i\alpha_{i+1}\beta_{i+1}\gamma_{i+1}\alpha_{i+1}\dots\alpha_{j-1}\beta_{j-1}\gamma_{j-1}
  \end{math}.
  From the construction above we have
  \begin{math}
  Q\yields{\lambda}Q_{i-1}\yields{\alpha_i} C_i\yields{\beta_i} D_i\yields{\lambda'}Q_{j-1}\yields{\alpha_j} C_j
  \end{math},  
  where
    $Q_{i-1}$ and $Q_{j-1}$ are quiescent,
    $C_i$ satisfies (i)-(iii) from the inductive hypothesis, and
    $\sig(C_i)=\sig(C_j)$.
  This situation, as well as the following construction is depicted in \cref{fig:lem:caslb_main}. 

  Now we apply \cref{lem:tst:hide_process} to configuration $C_i$.
  For the purpose of applying this lemma, we may assume that in $C_i$ process $p_{k+1}$ has just invoked a \WRead{} operation but not yet executed its first shared memory step during that operation.
  Hence, in all executions starting from $C_i$, $p_{k+1}$ will finish that pending \WRead{} operation in at most $t$ steps. 
  Then \cref{lem:tst:hide_process} yields a $P_{k+1}$-only schedule $\sigma_i$ such that for $D_i'=\Conf(C_i,\sigma_i)$
  \begin{equation}\label{eq:tst:500}
  D_i\indist_{\PP\setminus\set{p_{k+1}}} D_i',
  \end{equation}
  and one of the Cases~(a)-(c) of \cref{lem:tst:hide_process} hold.
  Let $C_j'=\Conf(D_i',\lambda'\alpha_j)$, and $D_j'=\Conf(C_j',\beta_i)$.
  (The use of $\beta_i$ instead of $\beta_j$ is intentional.)
  Then 
  \begin{math}
  Q_{i-1}\yields{\alpha_i} C_i\yields{\sigma_i}
  D_i'\yields{\lambda'\alpha_j} C_j'\yields{\beta_i} D_j'.
  \end{math}
  Since $\lambda'\alpha_j$ does not contain $p_{k+1}$, which is the only process that, according to \cref{eq:tst:500}, may be able to distinguish $D_i$ from $D_i'$, we obtain
  \begin{equation}\label{eq:tst:600}
  C_j'\indist_{\PP\setminus\set{p_{k+1}}} C_j.
  \end{equation}
  Configurations $C_i$ and $C_j$ have the same signature.
  Therefore, every process is poised to execute the same step in $C_i$ as in $C_j$, and all objects have the same values in both configurations.
  This, together with the fact that $\beta_i$ is $P_k$-only and every process appears at most once in $\beta_i$ implies
  \begin{equation}\label{eq:tst:700}
  \Exec(C_i,\beta_i)=\Exec(C_j,\beta_i)\stackrel{(\ref{eq:tst:600})}=\Exec(C_j',\beta_i).
  \end{equation}
  
  Hence, all objects have the same value in $D_i=\Conf(C_i,\beta_i)$ as in $D_j'=\Conf(C_j',\beta_i)$, and thus from \cref{eq:tst:500},
    all objects have the same value in $D_i'$ as in $D_j'$.
  Since $p_{k+1}$ does not appear in $\lambda'\alpha_j\beta_i$, and thus takes no step in the execution leading from $D_i'$ to $D_j'$, we conclude
  \begin{equation}\label{eq:tst:800}
  D_i'\indist_{p_{k+1}} D_j'.
  \end{equation}

  Now recall that $\sigma_i$ is the schedule $\sigma$ guaranteed by \cref{lem:tst:hide_process} (applied with $C=C_i$ and $q=p_{k+1}$), and the claim guarantees one of three Cases (a)-(c).
  First, assume Case~(a) occurs, i.e., $p_{k+1}$ completes a \WRead{} method call in $\Exec(C_i,\sigma_i)$ 
  (recall that in $C_i$ it had just invoked that method call)
  and is idle in $D_i'=\Conf(C_i,\sigma_i)$.
  Since process 0 takes no steps in $\Exec(C_i,\sigma_i)$ it follows that $D_i'$ is $p_{k+1}$-clean.
  On the other hand, $\Exec(D_i',\lambda'\alpha_j\beta_i)$ contains no steps by $p_{k+1}$, but instead a complete \WWrite{} by process 0.
  Hence, $D_j'=\Conf(D_i',\lambda'\alpha_j\beta_j)$, is $p_{k+1}$-dirty.
  But this contradicts \cref{obs:clean-dirty-distinguishability}, because according to \cref{eq:tst:800} process $p_{k+1}$ cannot distinguish $D_i'$ from $D_j'$.
  Hence, we know that Case~(a) from \cref{lem:tst:hide_process} cannot apply.
  
  Now, suppose that instead Case~(b) or (c) applies.
  We show that statments (i)--(iii) of the lemma are true for
    $\alpha=\lambda\alpha_i\sigma_i\lambda'\alpha_j$
    and
    $C=\Conf(Q,\alpha)=C_j'$.
  \oldpwtodo{The use of $C$ is not very good, because we used $C$ when we applied \cref{lem:tst:hide_process}.}
  By the inductive hypothesis (i), in $C_j$ all processes in $\PP\setminus P_k$ are idle, so from \cref{eq:tst:600} it follows that in $C_j'$ all processes in $\PP\setminus P_{k+1}$ are idle.
  This proves (i).
  
  According to Cases~(b) and (c) of \cref{lem:tst:hide_process}, in configuration $D_i'$ (and thus also in $C_j'$ and $D_j'$) process $p_{k+1}$ is poised to execute an operation $op$ that is either a \Write{} or a \CAS{$x$,$y$} on some object $R^\ast$.
  Moroever, 
  in case that $op$ is a \CAS{$x$,$y$}, in configuration $D_i'$ object  $R^\ast$ has value $x$.
  Then \cref{eq:tst:800} implies that the value of $R^\ast$ is also $x$ in configuration $D_j'$, and in partiular, if $p_{k+1}$ executes \CAS{$x$,$y$} in that configuration, that \CAS{} succeeds.
  We conclude that in the execution $\Exec(C_j',\beta_i\circ p_{k+1})$ process $p_{k+1}$ takes exactly one step, which is either a \Write{} or a successful \CAS{$x$,$y$}.
  By construction, $\Exec(C_i,\beta_i)$ is $P_k$-successful, and so \cref{eq:tst:700} implies that $\Exec(C_j',\beta_i)$ is also $P_k$-successful.
  It follows that $\Exec(C_j',\beta_i\circ p_{k+1})$ is $P_{k+1}$-successful, which proves statement (ii).

  Finally, since in $C_j'$ process $p_{k+1}$ is poised to execute operation $op$ on $R^\ast$, and all other processes are poised to execute exactly the same step as in $C_i$, we have: 
  In case $op$ is a \Write{},
  $\WCov(C_j',R^\ast)=\WCov(C_i,R^\ast)\cup\set{p_{k+1}}$, and 
  in case $op$ is a \CAS{}, $\CCov(C_j',R^\ast)=\CCov(C_i,R^\ast)\cup\set{p_{k+1}}$.
  All other sets $\WCov(\cdot,\cdot)$ and $\CCov(\cdot,\cdot)$ are the same for $C_i$ as for $C_j'$.
  Therefore, Cases~(b) and (c) of \cref{lem:tst:hide_process} together with the inductive hypothesis (iii) immediately imply  
  $|\WCov(C_j',R)\cap P_{k+1}|\leq t$ and $|\CCov(C_j',R)\cap P_{k+1}|\leq t$ for all objects $R$.
  This proves (iii) and completes the inductive step.
\end{proof}

Parts~(b) and (c) of \cref{thm:mainlb:ABA_detecting} follow immediately from this lemma.
(See \cref{sec:proof:mainlb:ABA_detecting} for additional details.)
Replacing each \WWrite{} with a \LL{}/\SC{} pair by process 0, and accommodating the definition of $p$-clean and $p$-dirty configurations, we obtain \cref{cor:mainlb:LLSC}.
(See \cref{sec:lb_LLSC} for additional details.)

%% file: aba_detection_ub.tex
\section{Upper Bounds}
\subsection{Constant-Time ABA-Detecting Registers from Registers}
\input{algorithm_with_CAS}
\cref{Fig:ABA-DetectingRegister} depicts an optimal linearizable implementation of an ABA-detecting register from $n+1$ bounded registers with constant step-complexity.
We use two more registers than needed according to the lower bound in \cref{thm:mainlb:ABA_detecting}\,(a).

\input{algorithm-v2}  

The main idea of the algorithm is similar to one used in the multi-layered construction of LL/SC/VL from CAS by Jayanti and Petrovic \cite{JP2003a}, which itself is a modified version of the implementation by Anderson and Moir~\cite{AM1995a}.

Here, we briefly discuss the implementation.
A complete correctness proof is provided in \cref{sec:proof:ABAReg-from-reg}.
We use a shared bounded register $X$ that stores a triple $(x,p,s)$, where $x$ is the value stored in the ABA-detecting register, $p\in\PP$ is a process ID, and $s\in\set{0,\dots,2n+1}$ is a sequence number.
We also use a shared \emph{announce array} $A[0\cdots n-1]$, where only process $q$ can write to $A[q]$.
Each array entry $A[q]$ stores a pair $(p,s)$, where $p\in\PP$ is a process ID and $s\in\set{0,\dots,2n+1}$ is a sequence number.
Register $X$ is initialized to $(\bot,\bot,\bot)$ and all entries of $A$ are initialized to $(\bot,\bot)$.

In a \DWrite{$x$} operation, the calling process $p$ first determines a suitable sequence number, $s$, using the helper method \GetSeq{}, and writes the pair $(x,p,s)$ to $X$ (\crefrange{l:Write:getSeq}{l:Write:writeX}).
Method \GetSeq{} ensures that the sequence number $s$ it returns satisfies the following:
If there is any point at which $X=(\cdot,p,s)$ and $A[q]=(p,s)$ for some process $q$, then $p$ will not use sequence number $s$ again in any following \DWrite{} call, until $A[q]\neq (p,s)$.
To achieve that, in a sequence of $n$ consecutive \GetSeq{} calls process $p$ scans through the entire announce array, reading one array entry with each \GetSeq{} call.
It then returns a sequence number that $p$ has not used in its preceding $n$ \DWrite{} method calls, and which it has not found in any array entry of $A[]$, when it read that entry last.

For its \DRead{} operation each process $q$ uses a local variable, $b$, that indicates whether a \DWrite{} operation linearized already during $q$'s previous \DRead{} operation after that operation's linearization point.
Our algorithms ensure that if $b=\True$ at the beginning of a \DRead{} operation, then such a \DWrite{} has happened.
In a \DRead{} operation, a process reads $X$ twice to obtain the triples $(x,p,s)$ in \cref{l:Read:readX}, and $(x',p',s')$ in \cref{l:Read:readX'}.
Between those two reads $q$ first reads the old value of $A[q]$ into $(r,s_r)$ (\cref{l:Read:readA[q]}), and then it \emph{announces} the pair $(p,s)$ by writing it to $A[q]$ (\cref{l:Read:writeA[q]}).
Now $(r,s_r)$ stores the ``old'' announcement from $q$'s preceding \DRead{} operation, and $A[q]$ stores the current one.
In \crefrange{l:Read:sameProcessSeq}{l:Read:ret=xTrue}, $q$ now decides the return value:
If $b=\True$ or if $(p,s)\neq (r,s_r)$, then $q$ returns $(x,\True)$, and otherwise $(x,\False)$.
Moreover, in preparation for the next \DRead{}, if $(x,p,s)=(x',p',s')$, $q$ sets $b$ to \False, and otherwise it sets it to \True.

First suppose $q$ reads two different triples from $X$ in \cref{l:Read:readX} and \cref{l:Read:readX'}, i.e., $(x,p,s)\neq (x',p',x')$.
Then the \DRead{} operation will linearize with the first read of $X$ (\cref{l:Read:readX}).
We now know that the value of $X$ has changed between the linearization point and the response of $q$'s \DRead{}.
Hence, $q$ sets flag $b$ to indicate that its next \DRead{} should return a pair $(\cdot,\True)$.
If $(x,p,s)=(x',p',s')$, then, on the other hand, it is ensured that $A[q]=(p,s)$ at the point when $q$ read $(x,p,s)$ from $X$ in \cref{l:Read:readX'}.
As explained above, in this case the pair $(p,s)$ will not be used again in any following \DWrite{} operation, until $q$ has replaced its announcement $(p,s)$ with a new one.
Hence, $q$ resets $b$ because in the following \DRead{} operation, $q$ will be able to detect any \DWrite{} that has happened inbetween by comparing $A[q]$ with the corresponding pair stored in $X$.

\subsection{LL/SC/VL from a Single Bounded CAS}
We now briefly sketch our wait-free implementation of LL/SC/VL from a single bounded CAS object.
The implementation has $O(n)$ step complexity, and thus, by \cref{cor:mainlb:LLSC}, is optimal.
The pseudo-code is presented in \cref{fig:LL/SC/VL-CAS} and correctness proofs can be found in \cref{sec:proof:LLSCVLfromCAS}.

In a CAS object $X$, we store a pair $(x,a)$, where $x$ represents the value of the implemented LL/SC/VL object, and $a$ is an $n$-bit string.
The $p$-th bit of $a$ is used to indicate whether an \SC{} operation linearized since $p$'s last \LL{} (the bit is usually set in this case).
As in the previous algorithm, we use a local variable $b$ for each process $p$.
In an \LL{} call, a process $p$ tries to reset \emph{its} bit (the $p$-th bit of the second component) of $X$.
As we explain below, this may fail, but only if an \SC{} $sc$ linearizes during $p$'s attempts to reset that bit.
If that happens, $p$ sets the flag $b$ and its \LL{} linearizes before $sc$.
Thus, in a subsequent \SC{} or \VL{}, $p$ determines from the set flag $b$ that it does not have a valid link, and that \SC{} or \VL{} can fail, even though $p$'s bit in $X$ is not set.

More precisely, in a \LL{} method call, process $p$ reads the pair $(x,a)$ from $X$ (\cref{l:LL:readX}) and checks whether its bit in $a$ is set.
If not, in \crefrange{l:LL:setBFalse1}{l:LL:returnX1} it simply resets $b$ (because in subsequent \SC{} or \VL{} calls $p$'s bit in $X$ will indicate whether $p$ has a valid link or not) and returns $x$.
That \LL{} operation linearizes with the \Read{} of $X$ in \cref{l:LL:readX}.
Now suppose $p$'s bit in $X$ is set.
Then $p$ tries to reset that bit, using a \CAS{} operation on $X$.
However, that \CAS{} may fail because of some other process' successful \CAS{} during a \LL{} or \SC{} call.
Therefore, $p$ repeatedly reads $X$ followed by a \CAS{} to set its bit in the second component of $X$, until its \CAS{} operation succeeds, or until it has failed $n$ times (\crefrange{l:LL:readX'}{l:LL:CAS}).
If a \CAS{} succeeds, $p$ resets $b$ and returns the first component of $X$ that it read just before its last, successful \CAS{} attempt (\crefrange{l:LL:setBFalse2}{l:LL:returnX'}); the \LL{} linearizes with that \CAS{}, and since $p$'s bit in $X$ is now reset, in the next \DRead{} operation $p$ can use its bit in $X$ to determine whether an \SC{} linearized since the linearization point of the current \DRead{}.
If $p$'s \CAS{} fails $n$ times, then $X$ must have changed $n$ times since $p$'s first \Read{} of $X$.
We argue that then at least one such change must be due to a \CAS{} operation during some process' \SC{}:
Suppose not.
Then $X$ must have changed at least $n$ times, and every time it must have changed because of a \CAS{} executed in a \LL{} operation.
But this is not possible, because each time such a \CAS{} succeeds, one of the bits in the second part of $X$ changes from 1 to 0, and $p$'s bit does not change at all.
We conclude that at least once, while $p$ has been trying to reset its bit in $X$, a successful \CAS{} on $X$ must have occurred during an \SC{} operation.
As we discuss below, this means that a successful \SC{} linearized.
Hence, in this case $p$ can set its bit to \True (\cref{l:LL:setBTrue}), which guarantees that $p$'s next \SC{} or \VL{} will fail, and return in \cref{l:LL:returnX2} the value $x$ it read at the very beginning from $X$.
The linearization point of that \LL{} is the \Read{} of $X$ in \cref{l:LL:readX}.

In an \SC{$y$} operation a process $p$ first checks flag $b$, and if it is set, $p$ immediately returns \False---this indicates that an \SC{} linearized during $p$'s last \LL{} but after the linearization point of that \LL{}.
If $b$ is not set, then $p$ reads $X$ to determine whether its bit in $X$ is set, and if yes, it can also return \False (\crefrange{l:SC:readX}{l:SC:returnFalse2}), because this indicates that some other \SC{} has linearized since $p$'s last \LL{}.
If $p$'s bit in $X$ is not set, then $p$ tries to write $(y,2^n-1)$ into $X$ using a \CAS{} operation
(\cref{l:SC:CAS}).
If that \CAS{} succeeds, as a result the value of the LL/SC/VL object change to $y$, and the bits of all processes are now set in $X$.
Hence, $p$'s \SC{$y$} linearizes with that successful \CAS{}, and $p$ returns \True (\cref{l:SC:returnTrue}).
If the \CAS{} fails, then $p$ repeats up to $n$ times, until either it finds that its bit in $X$ is set (and thus some other process' \SC{} succeeded), or its own \CAS{} succeeds.
If $p$'s \CAS{} fails $n$ times, then for the same reasons as explained earlier, we know that some process' \SC{} must have linearized during $p$'s ongoing \SC{$y$} operation, and thus $p$ can return \False (the unsuccessful \SC{} linearizes with its response).

Operation \VL{} is very simple: A process simply checks whether flag $b$ or its bit in $X$ is set, and if yes, it returns \False, otherwise it returns \True.

%% file: algorithm_with_CAS.tex
\begin{classfigure}[t!] 
\begin{minipage}{.47\textwidth}\vskip1ex
	\begin{method}{SC$_p$($x$)} 
		\lIf {$b$}{\Return \False\label{l:SC:returnFalse1}}

		\For {$i\leftarrow 1$ \KwTo $n$}
		{	
			$(y,a)\leftarrow X.\Read{}$\label{l:SC:readX}\;		
			\If (\tcc*[f]{if $p$'s bit is 1}){$\floor{a/2^p}$ is odd\label{l:SC:ifBitIsSet}}
			{
				\Return \False \label{l:SC:returnFalse2}
			}
			\If {$X.\CAS{$(y,a),(x,2^n-1)$}$\label{l:SC:CAS}}{\Return \True \label{l:SC:returnTrue}}
		}
		\Return \False \label{l:SC:returnFalse3}	
	\end{method} 
	\vspace{2pt}
	\begin{method}{VL$_p$()} 
		$(x,a)\leftarrow X.\Read{}$ \label{l:VL:readX}\;
		\eIf (\tcc*[f]{if $p$'s bit is 0 and $b=\False$}){($\floor{a/2^p}$ is even $\wedge$ $b=\False$)}
		{
			\Return $\True$\label{l:VL:returnTrue}
		}		
		{
			\Return $\False$\label{l:VL:returnFalse}
		}			
	\end{method} 		
\end{minipage}
\hspace{.01\textwidth}
\begin{minipage}{.5\textwidth}       
	\shared:
	\texttt{CAS} $X$\\ 
    	\local:
	\texttt{Boolean} $b$\\ 
	\begin{method}{LL$_p$()} 
		$(x,a)\leftarrow X.\Read{}$\label{l:LL:readX}\;
		\eIf (\tcc*[f]{if $p$'s bit is 0}){$\floor{a/2^p}$ is even}
		{
			$b \leftarrow \False$\label{l:LL:setBFalse1}\;
			\Return $x$\label{l:LL:returnX1}
		}		
		{
			\For{$i\leftarrow 1$ \KwTo $n$}
			{
				$(x',a') \leftarrow X.\Read{}$\label{l:LL:readX'}\;
				\If (\tcc*[f]{try to reset $p$'s bit}){$X.\CAS{$(x',a'),(x',a' - 2^p)$}$\label{l:LL:CAS}}
					{
						$b\leftarrow \False$\label{l:LL:setBFalse2}\;
						\Return $x'$\label{l:LL:returnX'}
					}
			}	
			$b\leftarrow \True$\label{l:LL:setBTrue}\;
			\Return $x$\label{l:LL:returnX2}
		}
		
	\end{method} 
\end{minipage}
\caption{An LL/SC/VL implementation from bounded CAS.\label{fig:LL/SC/VL-CAS}}
\label{Fig:LL/SC/VL-CAS}
\vspace{-15pt}
\end{classfigure}

%% file: algorithm-v2.tex
\begin{classfigure}[!t]
\begin{minipage}{.58\textwidth}       
    \vspace{5pt}

	\shared\\
	\texttt{Register} $X=(\bot,\bot,\bot)$\\ 
	\texttt{Register} $A[0\dots n-1]=((\bot,\bot),\dots,(\bot,\bot))$\;
    
    \vspace{5pt}
	\begin{method}{DWrite$_p$($x$)} 
		$s\leftarrow$ \GetSeq{}\label{l:Write:getSeq}\;
		$X.\Write{x,p,s}$\label{l:Write:writeX}\;
	
	\end{method} 

\vspace{3pt}

     	\begin{method}{GetSeq$_p$()} 

		$(r,s_r) \leftarrow A[c].\Read{}$\label{l:GetSeq:readA}\;
		\eIf{$r=p$\label{l:GetSeq:updateAnn}}
		{
			$na \leftarrow (na\setminus \{(c,i)\mid i \in \mathbb{N}\}) \cup (c,s_r)$\label{l:GetSeq:addtoAnn}
		}
		{
			$na \leftarrow na\setminus \{(c,i)\mid i \in \mathbb{N}\}$\label{l:GetSeq:removeFromAnn}
		}
		$c\leftarrow (c+1) \bmod n$\;

		choose arbitrary $s\in \{0,\dots,2n+1\}\setminus\big(\{i \mid (j,i)\in na\}\cup usedQ\big)$\label{l:GetSeq:chooseS}\;

		$usedQ.\enq{$s$}$\label{l:GetSeq:enqS}\;
		$usedQ.\deq{}$\label{l:GetSeq:deqS}\; 
		
		\Return $s$\;
	\end{method} 
\end{minipage}
\hspace{.015\textwidth}
\begin{minipage}{.4\textwidth}       

    	\local{(\textit{to each process})}\\
		\texttt{Boolean} b = \False\\ 
    		\texttt{Queue} $usedQ[n+1]=(\bot,\dots,\bot)$\\
		\texttt{Set} $na = \{\}$\\  
		\texttt{int} $c=0$\\  
\vspace{-5pt}
	\begin{method}{DRead$_q$()} 
		$(x,p,s)\leftarrow X.\Read{}$\label{l:Read:readX}\;
		$(r,s_r) \leftarrow A[q].\Read{}$\label{l:Read:readA[q]}\;
		$A[q].\Write{p,s}$\label{l:Read:writeA[q]}\;
		$(x',p',s')\leftarrow X.\Read{}$\label{l:Read:readX'}\;
		\eIf {$(p,s)=(r,s_r)$\label{l:Read:sameProcessSeq}}
    {
      $ret\leftarrow (x,b)$\label{l:Read:ret=xb}
     }		
     {
       $ret\leftarrow (x,\True)$\label{l:Read:ret=xTrue}
      }
		\eIf{$(x,p,s)=(x',p',s')$\label{l:Read:X=X'}}
		{
			$b\leftarrow \False$\label{l:Read:setbFalse}\;
		}		
		{
			$b\leftarrow \True$\label{l:Read:setbTrue}\;
		}
		\Return $ret$\;
	\end{method} 
\end{minipage}

\caption{An ABA-detecting register implemented from bounded registers.}
\label{Fig:ABA-DetectingRegister}

\end{classfigure}

%% file: appendix.tex
\section{Implementation of ABA-Detecting Registers from LL/SC/VL}
\label{sec:LL/SC/VL->ABA-Detecting}\begin{classfigure}[b!]
  \begin{minipage}[t]{.5\textwidth -0.25cm}       
    \shared
    \texttt{LL/SC Object} $X=\bot$\\ 
    
    \begin{method}{DWrite$_p$($x$)}
      $X$.\LL{}\label{l:ABALLSCVL:WriteLL}\;
      $X$.\SC{$x$}\label{l:ABALLSCVL:SC}\;
    \end{method} 
  \end{minipage}
  \hspace{.25cm}
  \begin{minipage}[t]{.5\textwidth -0.25cm}       
    \local
    $old=\bot$\\    
    \begin{method}{DRead$_q$()} 
      \IlIf{$X$.\VL{}\label{l:ABALLSCVL:VL}}{
        \Return{$(old,\False)$}
      }\;
      $old:=X.\LL{}$\label{l:ABALLSCVL:ReadLL};\quad \Return{$(old,\True)$}\;
    \end{method} 
  \end{minipage}
  \caption{Implementation of an ABA-detecting register from LL/SC/VL.
    For the ease of description but w.l.o.g.\ we assume that, even if $q$ has not called $X$.\FuncSty{LL()}, an $X$.\FuncSty{VL()} call by process $q$ returns \KwSty{True} as long as no successful $X$.\FuncSty{SC()} has been executed.
    \label{fig:ABA-Detecting-from-LL/SC/VL}}
\end{classfigure}
\begin{theorem}
  There is an implementation of an ABA-Detecting register from a single LL/SC/VL object, such that each \DRead{} and \DWrite{} operation takes only two shared memory steps.
\end{theorem}
\begin{proof}
  It suffices to show that the implementation in \cref{fig:ABA-Detecting-from-LL/SC/VL} is linearizable.
  
  Consider a history $H$ on the implemented ABA-Detecting register.  
  We assume w.l.o.g.\ that all operations in $H$ complete.
  
  We say a \VL{} operation \emph{succeeds}, if it returns \True.
  Note that we assumed w.l.o.g.\ (see the description of \cref{fig:ABA-Detecting-from-LL/SC/VL}) that the first $X$.\VL{} call by a process $q$ succeeds, even if $q$ has not called $X$.\LL{}, provided that no $X$.\SC{} call has been executed, either. 
  Therefore, for the purpose of this proof we may assume w.l.o.g.\ that the history $H$ starts with $n$ complete $X$.\LL{} operations, one by each process.  
  We say a process has a \emph{valid link} on $X$, whenever no successful $X$.\SC{} has occurred since $p$'s last $X$.\LL{} operation.

  For each \DWrite{} and \DRead{} operation $op$ by process $p$ respectively $q$, we define a point in time, $\lin{op}$, as follows.
  If $op$ is a \DWrite{} by process $p$, and $p$'s \SC{} in \cref{l:ABALLSCVL:SC} is successful, then $\lin{op}$ is the point when that successful \SC{} gets executed.
  If the \SC{} is unsuccessful, then $\lin{op}$ is the point immediately before the first successful \SC{} that gets executed after $p$'s \LL{} in \cref{l:ABALLSCVL:WriteLL} (such an \SC{} must occur before $p$'s \SC{} in \cref{l:ABALLSCVL:SC}, because that \SC{} by $p$ fails).
  If $op$ is a \DRead{} by $q$, then $\lin{op}$ is the point of the last shared memory operation executed by $q$ during $op$ (which is either $X$.\VL{} in \cref{l:ABALLSCVL:VL} or  $X$.\LL{} in \cref{l:ABALLSCVL:ReadLL}).
  
  We prove below that $\lin{}$ maps each operation $op$ to a linearization point of $op$; therefore, we say that $op$ linearizes at $\lin{op}$.
  Each point $\lin{op}$ occurs between the invocation and the response of $op$.
  It suffices to show that the history obtained by ordering all operations in $H$ by $\ell(op)$ is valid.
  Note that every \DWrite{} operation linearizes either at or immediately before the point of some successful \SC{}.
  Therefore, at any point the value of $X$ is equal to the value of the \DWrite{} operation that linearized last.
  
  Now consider a \DRead{} operation $op$ by process $q$.
  Initially, the value of $old$ equals the value of $X$, and $q$ has a valid link on $x$.
  It follows from the structure of the \DRead{} operation that the following invariant is maintained:
  At any point throughout $H$, $q$ has a valid link on $X$ if and only if no successful $X$.\SC{} has been executed on $X$ since $q$'s last $X$.\LL{} or its last successful $X.\VL{}$, whichever came later.
  Since a \DWrite{} operation linearizes at some point $t$ if and only if a successful $X$.\SC{} operation is executed at point $t$, and a \DRead{} operation linearizes with its successful $X$.\VL{} or, in case of an unsuccessful $X$.\VL{}, with its $X$.\LL{}, we obtain:
  Process $q$ has a valid link on $X$ if and only if no \DWrite{} has linearized since $q$'s last $X$.\DRead{} operation linearized (or since the beginning of the history, if none of $q$'s \DRead{} operations have linearized, yet).
  
  Now suppose $op$ linearizes with the $X$.\VL{} operation in \cref{l:ABALLSCVL:VL}, i.e., that \VL{} operation is successful.
  Then $q$ returns $(old,\False)$.
  The second component, \False, is correct because at $\lin{op}$ process $q$ has a valid link on $X$, i.e., no \DWrite{} operation has linearized since the linearization point of $q$'s preceding \DRead{} operation (or since the beginning of the execution, if $op$ is $q$'s first \DRead{} operation).
  The first component, $old$, is also correct, because $q$ has a valid link on $X$, which means that the value of $X$ cannot have changed since $q$'s last \LL{} in which it obtained the value of $old$ from $X$ (or since the beginning of the execution, when $X=old=\bot$).
  
  Finally, suppose $op$ linearizes with the $X$.\LL{} operation in \cref{l:ABALLSCVL:ReadLL}, i.e., the preceding $X$.\VL{} operation in \cref{l:ABALLSCVL:VL} fails.
  Then $op$ returns in \cref{l:ABALLSCVL:ReadLL}.
  Moreover, $q$ has no valid link at the point of that $X$.\VL{}, and thus also not immediately before $\lin{op}$.
  Hence, a \DWrite{} operation has linearized since the linearization point of $q$'s preceding \DRead{} operation (or since the beginning of the execution), and so the second component of the return value, \True, is correct.
  The first component of the return value is the value of $X$ at $\lin{op}$, which is also correct.  
\end{proof}

\section{Additional Details on the Lower Bound Proofs}\label{sec:lower_bound_additional_proofs}
\subsection{Proof of \cref{obs:clean-dirty-distinguishability}}\label{sec:proof:obs:clean-dirty}
\begin{proof}
  For the purpose of a contradiction, suppose $C_1\indist_p C_2$.
  Let $\alpha_i$, $i\in\set{1,2}$, be the schedules such that $C_{init}\yields{\alpha_i} C_i$, and
  \begin{itemize}  
    \item in $E_1:=\Exec(C_{init},\alpha_1)$ $p$ executes at least one complete \WRead{}, and the last one, $r^\ast$, happens after any \WWrite{}; and
    \item in $E_2:=\Exec(C_{init},\alpha_2)$ there is a complete \WWrite{} $w^\ast$ (by process 0) that overlaps with no \WRead{} by $p$, and $p$ invokes no \WRead{} after $w^\ast$.
  \end{itemize}
  
  By the nondeterministic solo-termination property, there exists a $p$-only schedule $\lambda$, such that in $\Exec(C_2,\lambda)$ process $p$ completes exactly one \WRead{} method call $r$. 
  Then in $E_2\circ \Exec(C_2,\lambda)$ the \WWrite{} $w^\ast$ happens before $r$ and after any preceding \WRead{} by $p$.
  Therefore, by the specification of \WWrite{} and \WRead{}, operation $r$ returns $\True$.
  
  Since $C_1\indist_p C_2$, process $p$ also completes its \WRead{} $r$ with return value \True during $\Exec(C_1,\lambda)$.
  But since $C_1$ is $p$-clean, there is now a complete \WRead{} operation $r^\ast$ by process $p$ that precedes $r$ in $\Exec(C_{init},\alpha_1\circ \lambda)$, and any \WWrite{} operation happens before $r^\ast$.
  By the specification of \WWrite{} and \WRead{}, operation $r$ returns \False---a contradiction.
\end{proof}

\subsection{Additional Details on the Proof of \cref{thm:mainlb:ABA_detecting}}\label{sec:proof:mainlb:ABA_detecting}
  Parts~(b) and (c) of \cref{thm:mainlb:ABA_detecting} follow almost immediately from \cref{lem:caslb_main}, as follows.
  
    By \cref{lem:caslb_main}, for $P=\set{1,\dots,n-1}$, there is a reachable configuration $C$ which has a $P$-successful schedule, and for every object $R$, $|\WCov(C,R)\cap P|\leq t$ and $|\CCov(C,R)\cap P|\leq t$.
    Hence, if the implementation uses $m$ writable CAS objects, then in $C$ at most $2t$ process are poised to access each CAS object, and thus
    \begin{displaymath}
    2mt\geq n-1.
    \end{displaymath}
    Similarly, if each of the $m$ base objects supports, in addition to \Read{}, only one of the two operations, \CAS{} or \Write{}, then in $C$ at most $t$ processes are poised to access each object, and so
    \begin{displaymath}
    mt\geq n-1.
    \end{displaymath}
    The result now follows immediately by solving for $m$.  

\subsection{Lower Bounds for Implementations of LL/SC}\label{sec:lb_LLSC}
We can easily modify the lower bounds for implementations of ABA-detecting registers to obtain the same lower bounds for implementations of LL/SC objects from bounded CAS objects and registers.

Consider a linearizable implementation of the methods \SC{} and \LL{}.
Given a process $p\in\PP\setminus\set{0}$, we define $p$-clean and $p$-dirty configurations in almost the same way as for \WRead{} and \WWrite{}:
Configuration $C$ is $p$-clean, if there exists a schedule $\alpha$, $C_{init}\yields{\alpha} C$, such that $\Exec(C_{init},\alpha)$ contains a complete \LL{} operation $r^\ast$ by $p$, and any (successful or unsuccessful) \SC{} happens before $r^\ast$.
Configuration $C$ is \emph{$p$-dirty}, if there exists a schedule $\alpha$, $C_{init}\yields{\alpha} C$, such that $\Exec(C_{init},\alpha)$ contains a successful \SC{} $w^\ast$, and no \LL{} by $p$ is pending at any point after $w^\ast$ has been invoked.

We observe analogously to \cref{obs:clean-dirty-distinguishability}:
\begin{observation}\label{obs:LL:clean-dirty-distinguishability}
	Suppose the \LL{} and \SC{} methods satisfy nondeterministic solo-termination.
	For any process $p\in\PP\setminus\set{0}$ and any two reachable configurations $C_1,C_2$, if $C_1$ is $p$-clean and $C_2$ is $p$-dirty, then $C_1\nindist_p C_2$ .
\end{observation}
\begin{proof}
  For the purpose of a contradiction, suppose $C_1\indist_p C_2$.
  Let $\alpha_i$, $i\in\set{1,2}$, be the schedules such that $C_{init}\yields{\alpha_i} C_i$, and
  \begin{itemize}  
    \item in $E_1:=\Exec(C_{init},\alpha_1)$ $p$ executes at least one complete \LL{}, and the last one, $r^\ast$, happens after any \SC{}; and
    \item in $E_2:=\Exec(C_{init},\alpha_2)$ there is a complete successful \SC{} $w^\ast$ that overlaps with no \LL{} by $p$, and $p$ invokes no \LL{} after $w^\ast$.
  \end{itemize}
  
  Then process $p$ must be idle in configuration $C_2$, and thus also in configuration $C_1$.
  The implementation of \LL{}/\SC{} must remain correct, if starting in configuration $C_1$ respectively $C_2$ process $p$ executes a \SC{$y$} call, where $y$ is an arbitrary value.
  In a long enough $p$-only execution, that \SC{$y$} call must complete, but since $C_1$ and $C_2$ are indistinguishable, the resulting $p$-only executions starting in $C_1$ respectively $C_2$ must be the same. 
  I.e., there is a $p$-only execution $E$ such that $E_1\circ E$ and $E_2\circ E$ are also executions, and in $E$ process $p$ completes exactly one \SC{} method $s^\ast$.
  In $E_1\circ E$, all \SC{} operations except for $s^\ast$ terminate before $p$'s last \LL{}, $r^\ast$, and that \LL{} is followed by $p$'s \SC{}, $s^\ast$.
  By the semantics of LL/SC, $s^\ast$ succeeds.
  On the other hand, in $E_2\circ E$ a successful \SC{} $w^\ast$ happens before $p$'s \SC{} $s^\ast$ and no \LL{} by $p$ is pending at any point after the invocation of $w^\ast$.
  Hence, $s^\ast$ fails, which is a contradiction.
\end{proof}

Now, in the proofs of \cref{lem:registerlb_main,lem:caslb_main}, we can simply replace every occurrence of \WRead{} 
with \LL{}
, 
and every occurrence of \WWrite{} with a \LL{} followed by an \SC{}.
With those replacements, every step of the proofs of \cref{lem:registerlb_main,lem:caslb_main} holds vacuously.
(Observe, that in the execution constructed in the original proofs, only process 0 calls \WWrite{}, so if we make the replacements as described, each \SC{} must succeeds.)
This yields \cref{cor:mainlb:LLSC}.
%

%% file: proof_of_ABA_reg_from_reg.tex
\section{Proof of \cref{Thm:ABAReg-from-Reg}}\label{sec:proof:ABAReg-from-reg}
To prove \cref{Thm:ABAReg-from-Reg} it is enough to show that the implementation of the ABA-detecting register in \cref{Fig:ABA-DetectingRegister} is linearizable.

In every line of the code, at most one shared memory operation is executed.
Consider some history $H$ in which processes execute \DRead{} and \DWrite{} operations.
For any operation $m$ by process $p$ and any line number $k$, we let $t^m_{k}$ denote the point in time when $p$ executes its shared memory operation in line $k$ during $m$.
Further, \inv{m} and \rsp{m} denote the points in time of the invocation respectively response of $m$.
We say operation $m$ \emph{completes} in some interval $[t,t']$, if $[\inv{m},\rsp{m}]\subseteq[t,t']$. 

Consider some history $H$ on the ABA-detecting register.
We define the linearization point $\lin{m}$ of each operation $m$ as follows.
A \DWrite{} operation $dw$ in $H$ linearizes when the value of $X$ is updated in $\cref{l:Write:writeX}$ of $dw$ (i.e.~$\lin{dw}=t^{dw}_{\ref{l:Write:writeX}}$).
For a \DRead{} operation $dr$ by some process $q$, we define $\lin{dr}=t^{dr}_{\ref{l:Read:readX}}$ if $b=\True$ at $\rsp{dr}$, and otherwise $\lin{dr}=t^{dr}_{\ref{l:Read:readX'}}$. 
Clearly the linearization point of each operation is between the invocation and response of the operation.
It remains to show that the history $S_H$ obtained by ordering all operations by their linearization points is valid.
For that, we first prove the following auxiliary claims.

\begin{claim}\label{Cl:invariant}
  For every complete $\DRead{}$ operation $dr$ in which process $q$ reads $(x,p,s)$ in \cref{l:Read:readX} and $(x',p',s')$ in \cref{l:Read:readX'},
   if $b=\True$ at \rsp{dr}, then some process writes to $X$ during $[\lin{dr},\rsp{dr}]$, and otherwise at \lin{dr}, we have  $A[q]=(p,s)=(p',s')$ and $(x,p,s)=(x',p',s')$.
\end{claim}

\begin{proof}
First suppose $b=\True$ at \rsp{dr}, and thus $\lin{dr}=t^{dr}_{\ref{l:Read:readX}}$.
This implies that $q$ executes \cref{l:Read:setbTrue} of $dr$, and so $(x,p,s)\neq(x',p',s')$.
I.e., $q$ reads different triples from $X$ in \cref{l:Read:readX} and in \cref{l:Read:readX'}.
Therefore, a process writes to register $X$ during $[t^{dr}_{\ref{l:Read:readX}},t^{dr}_{\ref{l:Read:readX'}}]\subseteq[\lin{dr},\rsp{dr}]$.

Now suppose $b=\False$ at \rsp{dr}, and thus $\lin{dr}=t^{dr}_{\ref{l:Read:readX'}}$.
It also implies that $q$ executes \cref{l:Read:setbFalse} of $dr$ and hence $(x,p,s)=(x',p',s')$.
Process $q$ writes the pair $(p,s)$ to $A[q]$ at $t^{dr}_{\ref{l:Read:writeA[q]}}$ and it does not change it before $t^{dr}_{\ref{l:Read:readX'}}$. 
Thus, $A[q]=(p,s)=(p',s')$ at $t^{dr}_{\ref{l:Read:readX'}}=\lin{dr}$.
\end{proof}

\begin{claim}\label{Cl:nGetSeq}
Consider two \GetSeq{} operations $gs_1$ and $gs_2$ by the same process $p$, which both return the same value $s$.
Then $p$ completes at least $n$ \GetSeq{} calls between $gs_1$ and $gs_2$.
\end{claim}

%
\begin{proof}
  This follows from the fact that before a process returns a sequence number $s$ in a \GetSeq{} call, it enqueues it in~\cref{l:GetSeq:enqS} in the queue $usedQ$ of length $n+1$.
  After that, according to \cref{l:GetSeq:chooseS}, it does not choose $s$ again until $s$ has been removed from the queue, and in every \GetSeq{} call only one element gets dequeued (in \cref{l:GetSeq:deqS}).
\end{proof}

\begin{claim}\label{Cl:announcedAtLinPoint}
Suppose  $X=(x,p,s)$ at some point $t$ for some triple $(x,p,s)$, and $A[q]=(p,s)$ throughout $[t,t']$, for some $t'\geq t$. 
Then, $p$ does not write $(x',p,s)$ into $X$ during $(t,t']$, for any value of $x'$.
\end{claim}
\begin{proof}
As $X=(x,p,s)$ at $t$, $p$ must have written that triple to $X$ before $t$ in a \DWrite{$x$} call.
Let $gs$ be the \GetSeq{} operation that $p$ executed during that \DWrite{$x$} call.
Then $gs$ responds before $t$ and returns $s$, and $p$ can complete at most one additional \GetSeq{} operation after $gs$ and before $t$ (this may happen during its following \DWrite{} call, just before it writes to $X$ again). 
Note that not more than one \GetSeq{} operation can be invoked after $gs$ and before $t$.

Suppose for the sake of contradiction that $p$ writes $(x',p,s)$ to $X$ during $(t,t']$ in \cref{l:Write:writeX} of some \DWrite{} operation, for some value $x'$.
Thus, $p$ completes a \GetSeq{} operation (in \cref{l:Write:getSeq} of the same \DWrite{} operation) during $[\rsp{gs},t']$, such that it returns $s$.
Let $gs'$ be the first such \GetSeq{} operation.

By \cref{Cl:nGetSeq}, $p$ completes at least $n$ \GetSeq{} operations $gs_1,\dots, gs_n$, executed in the same order, during $[\rsp{gs},\inv{gs'}]$.
As at most one \GetSeq{} operation can respond after $gs$ and before $t$, only $gs_1$ can get invoked and reponsd during $[\rsp{gs},t]$.
Thus, 
\begin{equation}\label{Eq:completeDuringtAndt'}
\text{$gs_2, \dots, gs_n, gs'$ all complete, in the same order, during $[t,t']$.}
\end{equation}
As $p$ increments its local variable $c$ by 1 modulo $n$ during each \GetSeq{} operation, $c=q$ at the invocation of some \GetSeq{} operation $gs''\in\{gs_2, \dots, gs_n, gs'\}$.
By the assumption, $A[q]=(p,s)$ throughout $[t,t']$, and therefore, by \cref{Eq:completeDuringtAndt'}, $A[q]=(p,s)$ throughout the execution of $gs''$.
Thus, $p$ reads $(p,s)$ from $A[q]$ and adds $(q,s)$ to its set $na$ in lines~\ref{l:GetSeq:readA}-\ref{l:GetSeq:addtoAnn} of $gs''$.
Process $p$ only removes $(q,s)$ from $na$, when it reads a new pair $(p,s')$, for $s'\neq s$, from $A[q]$ (lines~\ref{l:GetSeq:updateAnn}-\ref{l:GetSeq:removeFromAnn}), hence, $(q,s)$ remains in $p$'s set $na$ until some time after the value stored in $A[q]$ changes, which is after $t'$.
Therefore,
no \GetSeq{} operation by $p$ that executes \cref{l:GetSeq:chooseS} during  $[t^{gs''}_{\ref{l:GetSeq:addtoAnn}}, t']$ returns $s$.
But by \cref{Eq:completeDuringtAndt'}, $t_{\ref{l:GetSeq:chooseS}}^{gs'}\in [t^{gs''}_{\ref{l:GetSeq:addtoAnn}}, t']$, because $gs''\in\{gs_2, \dots, gs_n, gs'\}$ and $gs'$ is the last operation executed in this set.
Hence, $gs'$ does not return $s$, which contradicts the assumption.
%
\end{proof}

\begin{claim}\label{Cl:noWrite}
Consider two consecutive \DRead{} operations $dr_1$ and $dr_2$ by the same process $q$.
Suppose $b=\False$ at $\inv{dr_2}$, and the if-condition in \cref{l:Read:sameProcessSeq} of $dr_2$ evaluates to $\True$.
Then no process writes to $X$ during $[\lin{dr_1},\lin{dr_2}]$.
\end{claim}

\begin{proof}
Let $(x_1,p_1,s_1)$ and $(x_2,p_2,s_2)$ be the triples that process $q$ reads from $X$ in \cref{l:Read:readX} of $dr_1$ respectively $dr_2$.
As the value of $A[q]$ is only modified in \cref{l:Read:writeA[q]} of a \DRead{} operation by $q$, $A[q]=(p_1,s_1)$ throughout $(t^{dr_1}_{\ref{l:Read:writeA[q]}},t^{dr_2}_{\ref{l:Read:readA[q]}}]$, and so  $(r,s_r)=(p_1,s_1)$.
Since the if-condition in  \cref{l:Read:sameProcessSeq} of $dr_2$ evaluates to $\True$, $(p_1,s_1)=(p_2,s_2)$.
So let $p=p_1=p_2$ and $s=s_1=s_2$.

Then process $q$ writes $(p,s)$ to $A[q]$ at both $t^{dr_1}_{\ref{l:Read:writeA[q]}}$ and  $t^{dr_2}_{\ref{l:Read:writeA[q]}}$, and since $A[q]$ is not changed elsewhere, we have 
\begin{equation}\label{Lb:A[q]=(p,s)}
\text{$A[q]=(p,s)$ throughout $[t^{dr_1}_{\ref{l:Read:writeA[q]}},\rsp{dr_2}]$}.
\end{equation}

Suppose for the sake of contradiction that some process writes to $X$ during $[\lin{dr_1},\lin{dr_2}]$.
As $q$ reads $(x_2,p,s)$ from $X$ at $t^{dr_2}_{\ref{l:Read:readX}}$, the last write to $X$ during interval $[\lin{dr_1},t^{dr_2}_{\ref{l:Read:readX}}]$ must be by $p$ and it must write triple $(x_2,p,s)$ to $X$.
Process $q$ does not change the value stored in $b$ during $[\rsp{dr_1},\inv{dr_2}]$, and by the assumption $b=\False$ at $\inv{dr_2}$, thus, 
\begin{equation}\label{Lb:b=false}
b=\False \text{ at } \rsp{dr_1}.
\end{equation}
Thus, by the definition of $\ell$ we have 
\begin{equation}
\lin{dr_1}=t^{dr_1}_{\ref{l:Read:readX'}}.
\end{equation}
\Cref{Lb:b=false} also implies that if-condition in \cref{l:Read:X=X'} of $dr_1$ evaluates to \True.
Thus, $p$ reads the triple $(x_1,p,s)$ from $X$ at $t^{dr_1}_{\ref{l:Read:readX'}}$.
By~(\ref{Lb:A[q]=(p,s)}), we have $A[q]=(p,s)$ throughout $[t^{dr_1}_{\ref{l:Read:readX'}},\rsp{dr_2}]$, and so 
by Claim~\ref{Cl:announcedAtLinPoint}, $p$ does not write $(x',p,s)$ to $X$, for any value of $x'$, during $[t^{dr_1}_{\ref{l:Read:readX'}},\rsp{dr_2}]=[\lin{dr_1},\rsp{dr_2}]$, and so not during interval $[\lin{dr_1},t^{dr_2}_{\ref{l:Read:readX}}]$
---a contradiction.
\end{proof}

\begin{claim}\label{Cl:oneWrite}
Consider two consecutive \DRead{} operations $dr_1$ and $dr_2$ by some process $q$.
Suppose the if-condition in \cref{l:Read:sameProcessSeq} of $dr_2$ evaluates to $\False$.
Then a process writes to $X$ during $[\lin{dr_1},\lin{dr_2}]$.
\end{claim}

\begin{proof}
Suppose process $q$ reads some values $(x_1,p_1,s_1)$ and $(x_2,p_2,s_2)$  from $X$ in \cref{l:Read:readX} of $dr_1$ respectively $dr_2$.
Register $A[q]$ can only be modified by $q$ and only in \cref{l:Read:writeA[q]} of a \DRead{} operation, so $A[q]=(p_1,s_1)$ throughout $(t^{dr_1}_{\ref{l:Read:writeA[q]}},t^{dr_2}_{\ref{l:Read:readA[q]}}]$, and so  $(r,s_r)=(p_1,s_1)$.
By the assumption that the if-condition in  \cref{l:Read:sameProcessSeq} of $dr_2$ evaluates to $\False$, $(p_1,s_1)\neq(p_2,s_2)$.
Hence, the value $(x_2,p_2,s_2)$ gets written to $X$ during $(t^{dr_1}_{\ref{l:Read:readX}},t^{dr_2}_{\ref{l:Read:readX}})$.
If $\lin{dr_1}=t^{dr_1}_{\ref{l:Read:readX}}$, then the claim follows, as $\lin{dr_2}$ is either $t^{dr_2}_{\ref{l:Read:readX}}$ or $t^{dr_2}_{\ref{l:Read:readX'}}$.

Now suppose $\lin{dr_1}=t^{dr_1}_{\ref{l:Read:readX'}}$.
By the definition of $\ell$, $b=\False$ at $\rsp{dr_1}$.
Hence, $q$ executes \cref{l:Read:setbFalse}, and thus it reads the same triple $(x_1,p_1,s_1)$ from $X$ in \cref{l:Read:readX'}, as it did in~\cref{l:Read:readX} of $dr_1$.
Suppose for the sake of contradiction that no process writes to $X$ during $[\lin{dr_1},\lin{dr_2}]=[t^{dr_1}_{\ref{l:Read:readX'}},\lin{dr_2}]\supseteq[t^{dr_1}_{\ref{l:Read:readX'}},t^{dr_2}_{\ref{l:Read:readX}}]$.
Then $X$ remains unchanged throughout that interval, and in particular at $t_{\ref{l:Read:readX}}^{dr_2}$ process $q$ reads  $(x_1,p_1,s_1)$ from $X$, and so $(x_1,p_1,x_1)=(x_2,p_2,s_2)$---a contradiction. 
\end{proof}

Now, we prove that sequential history $S_H$ is valid. 
Consider the first \DRead{} $dr$ by some process $q$.
If no \DWrite{} linearizes before $t_{\ref{l:Read:readX'}}^{dr}$, then $X$ has its initial value, $(\bot,\bot,\bot)$, from the beginning of the execution until $t_{\ref{l:Read:readX'}}^{dr}$.
Hende, $q$ reads that triple from $X$ in \cref{l:Read:readX} and also in \cref{l:Read:readX'}, and so the if-condition in \cref{l:Read:X=X'} evaluates \True, and $d$ returns $(\bot,\False)$.
Since $\lin{dr}$ is before $t_{\ref{l:Read:readX'}}^{dr}$, and thus before any \DWrite{} linearizes, this return value is correct.
Now suppose some \DWrite{} operation linearizes before $t_{\ref{l:Read:readX'}}^{dr}$, and the last such operation uses parameter $x^\ast$.
If that happens before $t_{\ref{l:Read:readX}}^{dr}$, then $q$ reads a triple $(x^\ast,p,s)$ from $X$ in \cref{l:Read:readX}, where $(p,s)\neq(\bot,\bot)$.
But when $q$ executes \cref{l:Read:readA[q]}, $A[q]=(\bot,\bot)$, so $q$ assigns $ret$ the value $(x^\ast,\True)$ in \cref{l:Read:ret=xTrue}, and thus $dr$ later correctly returns that pair.
If the first \DWrite{} operation linearizes in $(t_{\ref{l:Read:readX}}^{dr},t_{\ref{l:Read:readX'}}^{dr})$, then $q$ reads $(\bot,\bot,\bot)$ from $X$ in \cref{l:Read:readX}, and $(x^\ast,p',s')$ in \cref{l:Read:readX'}, where $(p',s')\neq(\bot,\bot)$.
Hence, the if-condition in \cref{l:Read:sameProcessSeq} evaluates to \False, and $dr$ returns correctly $(x^\ast,\True)$.
 
Now suppose $dr$ is a \DRead{} by $q$, but it is not the first one.
For ease of notation, we write $dr_2$ instead of $dr$, and we let $dr_1$ be the \DRead{} by $q$ immediately preceding $dr_2$.
Let $(x^\ast,p^\ast,s^\ast)$ be the triple that $q$ reads from $X$ in \cref{l:Read:readX} of $dr_2$, and so $ret=(x^\ast,g)$ is the return value of $dr_2$, for some $g\in\{\True, \False\}$.
To prove that $S_H$ is a valid history, we show that
\begin{enumerate}[label=(\alph*)]
\item $X=(x^\ast,\cdot,\cdot)$ at $\lin{dr_2}$; and
\item $g=\True$ if and only if a \DWrite{} linearizes between $\lin{dr_1}$ and $\lin{dr_2}$.
\end{enumerate}

\paragraph{Proof of (a).} 
By definition, either $\lin{dr_2}=t^{dr_2}_{\ref{l:Read:readX}}$, or $\lin{dr_2}=t^{dr_2}_{\ref{l:Read:readX'}}$.
If $\lin{dr_2}=t^{dr_2}_{\ref{l:Read:readX}}$, then (a) is immediate, because $q$ reads $(x^\ast,p^\ast,s^\ast)$ from $X$ in that line.
If $\lin{dr_2}=t^{dr_2}_{\ref{l:Read:readX'}}$, then according to the definition of $\lin$, we have $b=\False$ at the response of $dr_2$.
Hence, $q$ executes \cref{l:Read:setbFalse} of $dr_2$, and thus it reads the same triple $(x^\ast,p^\ast,s^\ast)$ from $X$ in \cref{l:Read:readX'}.
It follows that $X=(x^\ast,\cdot,\cdot)$ when that happens, i.e., at $\lin{dr_2}=t^{dr_2}_{\ref{l:Read:readX'}}$.

\paragraph{Proof of (b).} 
First suppose $g=\False$.
This implies that \cref{l:Read:ret=xb} is executed during $dr_2$, and $b=\False$ at the invocation of $dr_2$.
Thus, by \cref{Cl:noWrite}, no process writes to $X$ during $[\lin{dr_1},\lin{dr_2}]$.

Now suppose $g=\True$.
Then either \cref{l:Read:ret=xb} of $dr_2$ is executed and $b=\True$ at the invocation of $dr_2$, or \cref{l:Read:ret=xTrue} of $dr_2$ is executed.
In the latter case, the if-condition in \cref{l:Read:sameProcessSeq} evaluates to \False, and so by \cref{Cl:oneWrite} a process writes to $X$ during $[\lin{dr_1},\lin{dr_2}]$.
Hence, consider the case that $b=\True$ at the invocation of $dr_2$.
Since $q$'s local variable $b$ does not change between consecutive \DRead{} method calls by $q$, we also have $b=\True$ at $\rsp{dr_1}$. 
Hence, by \cref{Cl:invariant}, a process writes to $X$ during $[\lin{dr_1},\rsp{dr_1}]\subseteq [\lin{dr_1},\lin{dr_2}]$.

%% file: LLSC_from_CAS_proof.tex
\section{Proof of \cref{Thm:LLSC-from-CAS}}\label{sec:proof:LLSCVLfromCAS}
\pwtodo{\cref{l:SC:returnFalse3} should be removed and instead the loop repeats $n+1$ times}
In this section, we prove \cref{Thm:LLSC-from-CAS} by showing that the implementation of LL/SC/VL object using CAS given in \cref{Fig:LL/SC/VL-CAS} is linearizable. 
Let \inv{m} and \rsp{m} denote the points in time of the invocation respectively response of some operation $m$.
First we show that if a process $p$ executes $n$ unsuccessful consecutive \CAS{} operations during a \LL{} or a \SC{} operation, then at least another process executes a successful \CAS{} during its \SC{} operation while the first process' \CAS{} operations fail.

\begin{claim}\label{Cl:nUnsuccessfulCAS}
Suppose a process $p$ executes n consecutive unsuccessful \CAS{} operations $c_1,\dots,c_n$ all either in \cref{l:LL:CAS} of a \LL{} or in \cref{l:SC:CAS} of a \SC{}.
Then during the time interval $I$ that starts when $p$ reads $X$ for the last time before $c_1$ (in \cref{l:LL:readX'}, respectively \cref{l:SC:readX}), and ends when $p$ finishes $c_n$, another process  executes a successful \CAS{} in~\cref{l:SC:CAS} of a \SC{} operation.
\end{claim}
\begin{proof}
Let $r_i$ be the \Read{} operation $p$ executes just before it executes $c_i$ (in \cref{l:LL:readX'}, or \cref{l:SC:readX}).
Operation $c_i$ fails if only if a process executes a successful \CAS{} operation between $p$'s $r_i$ and $c_i$.
As all $c_1,\dots,c_n$ fail, $n$ successful \CAS{} operations must have happened during interval $I$.

Now suppose for the sake of contradiction, that none of these $n$ successful \CAS{} operations during $I$ was due to a \CAS{} in \cref{l:SC:CAS}.
Hence, all $n$ successful \CAS{} operations during $I$ are due to a \CAS{} in \cref{l:LL:CAS}.
Each successful \CAS{} operation by some process $q$ in \cref{l:LL:CAS} resets $q$'s bit in the second component of $X$ to 0.
The second component of $X$ has $n$ bits, and each of these $n$ bits can change to 0 at most once, as no \CAS{} in \cref{l:SC:CAS} succeeds to change any of these bits to 1.
Moreover, none of $p$'s \CAS{} operations are successful.
Hence, at most $n-1$ successful \CAS{} in \cref{l:LL:CAS} can be executed during $I$---contradiction.
\end{proof}

To prove \cref{Thm:LLSC-from-CAS}, it suffices to prove that any history $H$ on the implementation of the LL/SC/VL object given in~\cref{Fig:LL/SC/VL-CAS} is linearizable.
For each operation $m$, we define the linearization point of $m$, \lin{m}, as follows.
For an \emph{unsuccessful} \SC{} operation $sc$ (i.e.~it returns \False in any of \cref{l:SC:returnFalse1,l:SC:returnFalse2,l:SC:returnFalse3}), we define $\lin{sc}=\rsp{sc}$.
A \emph{successful} \SC{} operation $sc$ (i.e.,~it returns \True in \cref{l:SC:returnTrue}) linearizes at the point at which its \CAS{} in \cref{l:SC:CAS} succeeds.
For a \VL{} operation $vl$ by some process $p$, if it returns $\False$ (in \cref{l:VL:returnFalse}), then $\lin{vl}=\rsp{vl}$.
If $vl$ returns $\True$ (in \cref{l:VL:returnTrue}), then $vl$ linearizes at the point when $p$ reads $X$ in \cref{l:VL:readX} of $vl$.
For a \LL{} operation $ld$ by some process $p$, we let \lin{ld} be the point at which $p$ executes \cref{l:LL:readX} if $ld$ returns in either \cref{l:LL:returnX1}, or \cref{l:LL:returnX2}.
If $ld$ returns in \cref{l:LL:returnX'}, then \lin{ld} is the point at which its \CAS{} in \cref{l:LL:CAS} succeeds.
It is not hard to see that the linearization point of each operation is between its invocation and its response.
It only remains to show that the sequential history $S_H$ obtained by ordering operations in $H$ by their linearization points is valid.
For that we first prove the following auxiliary claims.

\begin{claim}\label{Cl:aSCLinearizes}
Consider some \LL{} operation $ld$ by some process $p$, such that at \rsp{ld}, process $p$'s bit in $X$ is set or $b=\True$.
Then-and-only-then some successful \SC{} operation linearizes during $(\lin{ld},\rsp{ld}]$.
\end{claim}
\begin{proof}
First we prove the if-then statement.
The value of the local variable $b$ is updated during each \LL{} operation, just before the operation returns (\cref{l:LL:setBFalse1,l:LL:setBFalse2,l:LL:setBTrue}).
First consider the case at which $b=\True$ at \rsp{ld}.
This case can only happen if $ld$ returns in \cref{l:LL:returnX2}.
In this case all $n$ \CAS{} operations in \cref{l:LL:CAS} are unsuccessful.
By \cref{Cl:nUnsuccessfulCAS}, some process executes a successful \CAS{} operation during a \SC{} operation $sc$, while $p$ executes its $n$ unsuccessful \CAS{}. 
As in this case, \lin{ld} is when $p$ executes \cref{l:LL:readX}, operation $sc$ linearizes at its successful \CAS{} operation during $(\lin{ld},\rsp{ld}]$.

Now, suppose $b=\False$, but $p$'s bit  in $X$ is set at \rsp{ld}.
If $ld$ returns in \cref{l:LL:returnX1}, then $p$'s bit in $X$ is not set when $p$ reads $X$ in \cref{l:LL:readX} at \lin{ld}.
However, by the assumption, $p$'s bit in X is set when $ld$ responds.
This bit is only set when a \CAS{} operation succeeds during a \SC{} operation.
Hence, a process executes a successful \CAS{} during a \SC{} operation (and thus its \SC{} linearizes) during $(\lin{ld},\rsp{ld}]$.
If operation $ld$ returns in \cref{l:LL:returnX'}, $p$'s last \CAS{} operation in \cref{l:LL:CAS} of $ld$ must have been successful, and so $p$'s bit in $X$ must have changed to 0.
But by the assumption, $p$'s bit is set at \rsp{ld}, so some other process must have changed it back to 1 after $p$'s successful \CAS{}.
As \lin{ld} is when $p$'s \CAS{} succeeds, the value of $p$'s bit changes after \lin{ld} and before \rsp{ld}.
Recall that $p$'s bit is only set when a process executes a successful \CAS operation during a \SC{}.
Therefore, some process must have executed a successful \CAS{} operation during a \SC{} and thus its \SC{} linearized during $(\lin{ld},\rsp{ld}]$.

Now we prove the only-then statement.
For that we show if $p$'s bit in $X$ is not set and $b=\False$ at \rsp{ld}, then no \SC{} linearizes during $(\lin{ld},\rsp{ld}]$.
Local variable $b=\False$ at \rsp{ld}, therefore $ld$ returns either in \cref{l:LL:returnX1} or in \cref{l:LL:returnX'}.
In the case that $ld$ returns in \cref{l:LL:returnX1}, $p$'s bit is 0 when $p$ reads $X$ in \cref{l:LL:readX} (at the linearization point of $ld$).
Thus, $p$'s bit is 0 at both \lin{ld} and \rsp{ld}.
This bit can only be changed to 0 by $p$ and only in \cref{l:LL:CAS}, which is not executed during $ld$ in this case.
Hence $p$'s bit has value 0 throughout $(\lin{ld},\rsp{ld}]$.
As all processes' bits in $X$ change to 1 when a \SC{} linearizes at its successful \CAS{} in \cref{l:SC:CAS}, no successful \CAS{} of a \SC{} happens throughout $(\lin{ld},\rsp{ld}]$.

Now suppose $ld$ returns in \cref{l:LL:returnX'}.
In this case $ld$ linearizes when its successful \CAS{} changes $p$'s bit to 0.
Process $p$'s bit is not set at \rsp{ld} and as $p$ does not try to change the value of its bit after its successful \CAS{}, and thus after \lin{ld}, $p$'s bit has value 0 throughout  $(\lin{ld},\rsp{ld}]$, and so with the same argument as before, no \SC{} linearizes at its successful \CAS{} operation in this interval.
\end{proof}

\begin{claim}\label{Cl:noSCBetweenReadAndCAS}
Consider a successful \CAS{} operation $cas$ in \cref{l:SC:CAS} of a \SC{} operation $sc$ and the last \Read{} operation $r$ executed before $cas$ in \cref{l:SC:readX} of $sc$. 
Then no successful \SC{} linearizes between $r$ and $cas$.
\end{claim}
\begin{proof}
Let $p$ be the process which executes $sc$, and $(y^\ast,a^\ast)$ be the value $p$ reads from $X$ when it executes $r$.
As $cas$ succeeds, the if-condition in \cref{l:SC:ifBitIsSet} cannot be evaluated to $\True$.
Hence, $p$'s bit in $X$ must be 0 when $p$ executes $r$, and so $a^\ast\neq 2^n-1$.
Moreover, since $cas$ is successful, the value of $X$ is $(y^\ast,a^\ast)$ just before $cas$ is executed.
Suppose for the sake of contradiction that at least one successful \SC{} operation linearizes between $r$ and $cas$.
Note that the value of $X$ is updated at the linearization point of a successful \SC{} operation.
Thus, the last successful \SC{} executed between $r$ and $cas$ must update the value of $X$ to $(y^\ast,a^\ast)$.
However, a successful \SC{} operation changes the second component of $X$ to $2^n-1$, and so $a^\ast=2^n-1$---contradiction.
\end{proof}

\begin{claim}\label{Cl:SCisValid}
Consider a \SC{} operation $sc$ by some process $p$ and let $ld$ be the last \LL{} operation by the same process $p$ executed before $sc$.
Then $sc$ is successful if and only if no successful \SC{} operation linearizes between $ld$ and $sc$.
\end{claim}
\begin{proof}
First we prove the if-then statement.
Operation $sc$ is successful, if one of its \CAS{} operations $c^*$ succeeds in \cref{l:SC:CAS} and so $sc$ returns in \cref{l:SC:returnTrue}.
Hence, $b=\False$ at the \inv{sc}.
As $p$'s local variable $b$ is only changed during a \LL{} operation, 
\begin{equation}\label{Eq:b=False}
\text{$b=\False$ at the \rsp{ld}.}
\end{equation}
Moreover, since $sc$ returns in \cref{l:SC:returnTrue}, $p$ reads value $0$ from its bit when it reads $X$ in \cref{l:SC:readX} the last time before $cas$ at some point $t$.
This bit can only be reset in \cref{l:LL:CAS} of a \LL{} operation by $p$, hence, 
\begin{equation}\label{Eq:bitIs0}
\text{$p$'s bit is 0 throughout  $[\rsp{ld},t]$.}
\end{equation}
Hence by \cref{Eq:b=False}, \cref{Eq:bitIs0}, and \cref{Cl:aSCLinearizes}, no successful \SC{} operation linearizes during $(\lin{ld},\rsp{ld}]$.
Moreover, by \cref{Eq:bitIs0}, no successful \SC{} linearizes throughout $[\rsp{ld},t]$, as otherwise the value of $p$'s bit would change to 1.
Moreover, by \cref{Cl:noSCBetweenReadAndCAS} no successful \SC{} linearizes during $[t,\lin{sc}]$, as \lin{sc} is when cas succeeds.
Therefore, no successful \SC{} linearizes throughout $(\lin{ld},\lin{sc}]$.

Now we show the only-if statement is also true, by showing that if $sc$ is not successful, then at least one successful \SC{} operation linearizes between $ld$ and $sc$.
There are three cases where $sc$ can return.
The first case is if $sc$ returns in \cref{l:SC:returnFalse1} and so $b=\True$ at \inv{sc}.
Process $p$'s local variable $b$ does not change outside a \LL{} operation, hence, $b=\True$ at the \rsp{ld}.
By \cref{Cl:aSCLinearizes}, a successful \SC{} operation linearizes during $(\lin{ld},\rsp{ld}]\subseteq (\lin{ld},\lin{sc}]$.

The second case happens when $sc$ returns in \cref{l:SC:returnFalse2}.
In this case, $p$'s bit is set when $p$ reads $X$ in \cref{l:SC:readX} for the last time during $sc$ at some point $t$.
Note that $\lin{sc}=\rsp{sc} \geq t$.
Now suppose $p$'s bit is 0 at \rsp{ld}.
Hence, some process sets this bit with a successful \CAS{} at the linearization point of a successful \SC{} operation during $(\rsp{ld},t]\subseteq(\lin{ld},\lin{sc}]$.
If $p$'s bit is 1 at \rsp{ld}, then by \cref{Cl:aSCLinearizes}, a successful \SC{} operation linearizes during  $(\lin{ld},\rsp{ld}]\subseteq(\lin{ld},\lin{sc}]$.

The last case is when $sc$ returns in \cref{l:SC:returnFalse3}.
This implies that all $n$ \CAS{} operations of $p$ during $sc$ failed.
Thus by \cref{Cl:nUnsuccessfulCAS}, a successful \CAS happens during $[\inv{sc},\rsp{sc}]$ and so a successful \SC{} linearizes during $(\lin{ld},\lin{sc}]$.
\end{proof}

\begin{claim}\label{Cl:VLisValid}
Consider a \VL{} operation $vl$ by some process $p$ and let $ld$ be the last \LL{} operation by the same process $p$ executed before $sc$.
Then $vl$ returns $\True$ if and only if no successful \SC{} operation linearizes between $ld$ and $sc$.
\end{claim}
\begin{proof}\todo{this is too similar to the previous claim, the proof can be removed if we need space}
First we prove the if-then statement.
Operation $vl$ returns $\True$, hence, $b=\False$ at the \inv{vl}.
As $p$'s local variable $b$ is only changed during a \LL{} operation, 
\begin{equation}\label{Eq:VL:b=False}
\text{$b=\False$ at the \rsp{ld}.}
\end{equation}
Moreover, since $vl$ returns $\True$, $p$ reads value $0$ from its bit when it reads $X$ in \cref{l:VL:readX} of $vl$ at the linearization point of $vl$.
This bit can only be reset in \cref{l:LL:CAS} of a \LL{} operation by $p$, hence, 
\begin{equation}\label{Eq:VL:bitIs0}
\text{$p$'s bit is 0 throughout  $[\rsp{ld},\lin{vl}]$.}
\end{equation}
Hence by \cref{Eq:VL:b=False}, \cref{Eq:VL:bitIs0}, and \cref{Cl:aSCLinearizes}, no successful \SC{} operation linearizes during $(\lin{ld},\rsp{ld}]$.
Moreover, by \cref{Eq:VL:bitIs0}, no successful \SC{} linearizes throughout $[\rsp{ld},\lin{vl}]$, as otherwise the value of $p$'s bit would change to 1.
Therefore, no successful \SC{} linearizes throughout $(\lin{ld},\lin{vl}]$.

Now we show the only-if statement is also true, by showing that if $vl$ returns $\False$, then at least one successful \SC{} operation linearizes between $ld$ and $vl$.
There are two cases that can cause $vl$ to return $\False$.
The first case is if $b=\True$ at \inv{sc}.
Process $p$'s local variable $b$ does not change outside a \LL{} operation, hence, $b=\True$ at the \rsp{ld}.
By \cref{Cl:aSCLinearizes}, a successful \SC{} operation linearizes during $(\lin{ld},\rsp{ld}]\subseteq (\lin{ld},\lin{vl}]$.
The second case is when $p$'s bit is set when $p$ reads $X$ in \cref{l:VL:readX} of $vl$ at some point $t$.
Note that $\lin{vl}=\rsp{vl} \geq t$.
Now suppose $p$'s bit is 0 at \rsp{ld}.
Hence, some process sets this bit with a successful \CAS{} at the linearization point of a successful \SC{} operation during $(\rsp{ld},t]\subseteq(\lin{ld},\lin{vl}]$.
If $p$'s bit is 1 at \rsp{ld}, then by \cref{Cl:aSCLinearizes}, a successful \SC{} operation linearizes during  $(\lin{ld},\rsp{ld}]\subseteq(\lin{ld},\lin{vl}]$.
\end{proof}

Now we can quickly argue why $S_H$ is valid.
Consider some \LL{} operation $ld$ in $S_H$ that returns some value $y$.
For $S_H$ to be valid, register $X$ must contain value $y$ at the linearization point of $ld$.
Recall that \lin{ld} is the point at which $p$ executes \cref{l:LL:readX}, if $ld$ returns in either \cref{l:LL:returnX1} or \cref{l:LL:returnX2}.
Moreover, if $ld$ returns in \cref{l:LL:returnX'}, then \lin{ld} is the point at which $p$'s \CAS{} in \cref{l:LL:CAS} succeeds.
In the former case, $y$ is the value that $p$ reads from $X$ in \cref{l:LL:readX}  at \lin{ld}.
In the latter case, $y$ is the value that $p$ writes into $X$ during its successful \CAS{} operation in \cref{l:LL:CAS} at \lin{ld}.
Hence $ld$ returns a valid value.
This in addition to the results of \cref{Cl:SCisValid} and \cref{Cl:VLisValid} complete the proof that the resulting history $S_H$ is valid.